\documentclass[11pt]{article}

\usepackage{booktabs} 

\usepackage{amsmath,amsthm,amsfonts,amssymb,fullpage}
\usepackage{dsfont}
\usepackage{color}
\definecolor{mylinkcol}{RGB}{0,30,70}
\usepackage{tikz,pgfplots}
\usetikzlibrary{shapes}
\usepackage[font=small]{caption}
\usepackage{enumitem}
\usepackage[numbers]{natbib}
\usepackage[colorlinks=true,breaklinks=true,bookmarks=true,urlcolor=blue,
     citecolor=blue,linkcolor=blue,bookmarksopen=false,draft=false]{hyperref} 

\usetikzlibrary{external}
\usepackage[colorinlistoftodos,textsize=tiny]{todonotes}
\newcommand{\Comments}{0} 
\newcommand{\FutureComments}{0} 
\definecolor{gray}{gray}{0.5}
\definecolor{darkgreen}{rgb}{0,0.5,0}
\definecolor{bo}{rgb}{0.5,0.5,0}
\definecolor{raf}{rgb}{0,0.5,0.5}
\newcommand{\mynote}[2]{\ifnum\Comments=1\textcolor{#1}{#2}\fi}
\newcommand{\mytodo}[2]{\ifnum\Comments=1%
  \todo[linecolor=#1!80!black,backgroundcolor=#1,bordercolor=#1!80!black]{#2}\fi}

\newcommand{\future}[1]{\ifnum\FutureComments=1{FUTURE: #1]}\fi}

\newtheorem{theorem}{Theorem}[section]

\newtheorem{lemma}{Lemma}[section]
\newtheorem{proposition}{Proposition}[section]

\newtheorem*{theorem*}{Theorem}
\newtheorem*{proposition*}{Proposition}

\theoremstyle{definition}
\newtheorem{definition}{Definition}
\newtheorem{axiom}{Axiom}

\newcommand{\interior}{\mathrm{int}}
\newcommand{\tr}{\top}

\newcommand{\E}{\mathbb{E}}
\renewcommand{\P}{\mathcal{P}}
\newcommand{\R}{\mathcal{R}}

\newcommand{\Y}{\mathcal{Y}}

\newcommand{\conv}{\mathrm{conv}}
\newcommand{\clo}{\mathrm{cl}}
\newcommand{\toto}{\rightrightarrows}

\newcommand{\defeq}{\doteq}
\newcommand{\ones}{\mathds{1}}

\newcommand{\reals}{\mathbb{R}}

\newcommand{\outcomes}{\mathcal{Y}}

\newcommand{\AnySrange}[1]{\mathcal{#1}}
\newcommand{\Srange}{\AnySrange{S}}
\newcommand{\Dc}{\mathcal{D}}
\newcommand{\Hc}{\mathcal{H}}
\newcommand{\Qc}{\mathcal{Q}}
\newcommand{\Xc}{\mathcal{X}}

\newcommand{\argmax}{\mathop{\mathrm{argmax}}}

\def\FUNF(#1,#2,#3){(\FUNS(#1,#3))-(\FUNS(#1,#2))}

\begin{document}
\title{An Axiomatic Study of Scoring Rule Markets}
\author{Rafael Frongillo \\CU Boulder \and Bo Waggoner\\UPenn}
\date{}
\maketitle

\begin{abstract}
Prediction markets are well-studied in the case where predictions are probabilities or expectations of future random variables. In 2008, Lambert, et al. proposed a generalization, which we call ``scoring rule markets'' (SRMs), in which traders predict the value of arbitrary statistics of the random variables, provided these statistics can be elicited by a scoring rule. Surprisingly, despite active recent work on prediction markets, there has not yet been any investigation into the properties of more general SRMs.
To initiate such a study, we ask the following question: in what sense are SRMs ``markets''? We classify SRMs according to several axioms that capture potentially desirable qualities of a market, such as the ability to freely exchange goods (contracts) for money. Not all SRMs satisfy our axioms: once a contract is purchased in any market for prediction the median of some variable, there will not necessarily be any way to sell that contract back, even in a very weak sense. Our main result is a characterization showing that slight generalizations of cost-function-based markets are the only markets to satisfy all of our axioms for finite-outcome random variables. Nonetheless, we find that several SRMs satisfy weaker versions of our axioms, including a novel share-based market mechanism for ratios of expected values.%
  
\end{abstract}

\section{Introduction}
\label{sec:msr-introduction}
The goal of a prediction market is to collect and aggregate predictions about some future outcome $Y$ taking values in $\outcomes$; common examples arise from sporting, political, meteorological, or financial events.
Prediction markets work by offering financial contracts whose payoffs are contingent on the eventually-observed value of $Y$.
The agent's choices are interpreted, by revealed preference, as predictions about $Y$, and thus the final state of the market is interpreted as an aggregation of agent beliefs.

Hanson~\citep{hanson2003combinatorial} observed that one can design a prediction market using a \emph{proper scoring rule}, which is a contract $S(p',y)$ that scores the accuracy of prediction $p'$ (a probability distribution over outcomes) upon outcome $Y=y$.  The relevant guarantee is that the expected score $\E_{Y\sim p} S(p',Y)$ is maximized when the agent reports their true belief $p'=p$.  In Hanson's \emph{scoring rule market (SRM)},\footnote{We will use this term in lieu of the standard \emph{market scoring rule (MSR)}, as the latter could refer to either the scoring rule or the market mechanism.} traders arrive sequentially and report their belief $p_t$, and are eventually paid $S(p_t,y)-S(p_{t-1},y)$ when $Y=y$ is revealed.

In many cases, a market designer may only be interested in specific (e.g.\ summary) statistics, called \emph{properties}, of trader beliefs, rather than the entire distribution over the outcome $Y$.
Lambert et al.~\citep{lambert2008eliciting} observed that SRMs can be generalized to arbitrary scoring rules $S(r,y)$, wherein traders are paid $S(r_t,y)-S(r_{t-1},y)$, and if the expected value of $S$ is maximized by reporting the value of a particular property, in which case we say $S$ \emph{elicits} the property, then the market should intuitively aggregate trader beliefs about the property in question.
For example, one could design a ``median market'' by leveraging the fact that $S(r,y) = -|r-y|$ elicits the median of the distribution of $Y$, and we might expect the corresponding SRM to aggregate predictions about the median of $Y$.

Surprisingly, apart from exploring which properties are elicitable, practically nothing is known about these general SRMs apart from one property: expected values~\citep{hanson2003combinatorial,chen2007utility,abernethy2013efficient}.
Here SRMs can be rephrased into a dual ``cost-function-based'' formulation where traders buy and sell \emph{shares} in some securities $\phi_1,\ldots,\phi_k \in \reals^\Y$, and the market prices reveal the trader's belief about the expected value of $\phi: \outcomes \to \reals^k$ (\S~\ref{sec:autom-mark-makers},~\ref{sec:char-tn}).
The literature on prediction markets focuses on this setting, in which one can more easily study traditional market quantities like liquidity and depth.

In this paper, we step beyond expected values and study SRMs as a whole.
Our first contribution is identifying which questions to ask; in particular we ask the following.
\begin{enumerate}[label=(\arabic*),itemsep=0pt,topsep=4pt]
\item In what sense can we think of SRMs as ``markets'' in the traditional sense?
\item Specifically, which SRMs behave like share-based markets, in the sense that traders can always ``sell'' their shares for some nontrivial price?
\item Which SRMs if any maintain the known desirable characteristics of the cost-function-based framework, such as bounded worst-case loss and no arbitrage?
\end{enumerate}
To answer these questions, we introduce axioms, and then study SRMs based on the property they elicit---as we will see, for example, markets eliciting modes will satisfy different axioms than those for medians, regardless of the scoring rules chosen.
Our primary focus is the \emph{trade neutralization (TN)} axiom, which captures question (2) above: traders holding a contract from the market maker can ``sell'' it for a nontrivial price, receiving more from transaction than the worst-case payoff of the contract.
(Otherwise, traders would just keep their holdings regardless of their belief.)

\paragraph{Summary and results.}
After reviewing prediction markets and scoring rules (\S~\ref{sec:prediction-market-overview}), we introduce our axioms in \S \ref{sec:axioms}.
We show that under some reasonably mild and standard axioms, incentive compatibility (IC) and path independence (PI), any mechanism that elicits predictions by offering contracts must be an SRM, thus justifying our focus on them.
We then classify modes, finite properties, medians, quantiles, and expectations, according to what axioms their markets can satisfy (\S~\ref{sec:first-examples}).

In \S~\ref{sec:char-tn}, we give our main result: \emph{any SRM satisfying trade neutralization must be a certain form of generalized cost-function-based market} (Theorem~\ref{thm:tn-ic-implies-cost-func}).
We note that unlike in prior work on cost functions~\citep{abernethy2013efficient}, this result applies to any SRM for any property; so in a sense, it says that only close relatives of expectations can be elicited by any market satisfying trade neutralization. 

Our exploration into non-expected-value SRMs yields a new prediction market mechanism worthy of attention on its own, for eliciting the ratio of expectations of random variables.
As we describe in \S~\ref{sec:ratios}, this new market framework exchanges a security (random variable) not for cash but for units of another security, thereby revealing trader beliefs about the \emph{ratio} of their expectations.%
\footnote{To see this, note that a trader willing to ``buy'' one unit of $d\in\reals^\Y$ in exchange for $c$ units of $b\in\reals^\Y$ is effectively expressing the belief $\E\,d > c \, \E\,b$, i.e., $\E\,d / \E\,b > c$.  Similarly, selling at this ``price'' reveals the belief $\E\,d / \E\,b < c$.}
Finally, we conclude in \S~\ref{sec:discussion} with a discussion of markets for other properties, elicitation complexity, and alternative market formulations.
The Appendix contains all omitted proofs.

\paragraph{Relation to prior work.}
While we discuss related work in \S~\ref{sec:prediction-market-overview}, it is important to distinguish our main result from prior characterizations of cost-function-based markets in the literature.
In particular, as mentioned above, cost-function-based markets are known to be equivalent to scoring rules that elicit expected values~\citep{hanson2003combinatorial,chen2007utility,abernethy2013efficient}.
In contrast, we ask which among a very broad family of mechanisms satisfy certain basic ``market'' axioms, and find that such mechanisms must fall into (a minor generalization of) the cost-function-based framework.
Indeed, based on the above known equivalence, we can further conclude that our axioms imply the elicitation of expected values, though weakening them slightly allows us to elicit ratios of expectations (\S~\ref{sec:ratios}) and expectiles (\S~\ref{app:expectiles}) while still retaining some sense of a ``market''.

\section{Aggregating Information with Prediction Markets}
\label{sec:prediction-market-overview}

The goal of a prediction market is to crowdsource and aggregate beliefs of participants about some future event.
It does so by allowing participants to buy and sell contracts which have different payoffs depending on the outcome of the event, and inferring a prediction from the participants' choices.
In this section, we formally define the class of such markets that we study, \emph{scoring rule markets (SRMs)}, with references to previous work.
For other related work, see above and \S~\ref{sec:discussion}.

\subsection{Outcomes and contracts}
\label{sec:outcomes-contracts}
Let $\Y$ denote the \emph{outcome space} of interest to the market designer and $Y$ a future event or random variable taking values in $\mathcal{Y}$.
The designer will in general allow traders to select \emph{contracts} $d \in \reals^{\outcomes}$ from a list offered by the market.
The interpretation of a contract $d$ is that, when $Y$ is eventually revealed, the market maker will pay the agent a net payoff $d(Y)$ (which may be negative).

Note that under this formulation, any initial payment the agent might make is folded in to $d$.  Specifically, letting $\ones\in\reals^\Y$ denote the ``all-ones'' contract $\ones(y) = 1 \,\forall y\in\Y$, then a contract $d$ could be written $d = d'-\alpha\ones$, interpreted as paying a price of $\alpha$ now for the contract $d'$, whose payoff will be revealed when $Y$ is observed.
We assume agents are indifferent to timing of payments and just wish to maximize total expected payoff.

\subsection{Automated market makers}
\label{sec:autom-mark-makers}
When designing prediction markets, rather than a typical continuous double auction (``stock exchange'') mechanism, it is common to employ a centralized \emph{automated market maker}, which offers to buy or sell any available contracts, and through which all trades are executed.  (See \cite{abernethy2013efficient} for practical reasons behind this design choice.)

Formally, a sequence of participants (traders) arrive at times $t=1,\dots,T$ and each selects a contract from a list offered by the market at that time.
It will be convenient to let the set of contracts available be indexed at each time by some \emph{report space} $\R \subseteq \reals^k$.
Following Abernethy et al.~\citep{abernethy2014general}, we consider a generic market making algorithm, termed a \emph{mechanism}, that specifies the set of contracts available at each time.
In general, this may depend on the entire past history of the market, and is represented as a mapping $F$ that, given a report $r \in \R$ of the participant, returns the corresponding contract.
More formally, the contract given to a participant who chooses report $r_t$ given the past history of reports $r_1,\ldots,r_{t-1}$ is denoted $\vec F(r_t|r_1,\ldots,r_{t-1}) \in \reals^\Y$.
The net payoff to the trader upon outcome $y\in\Y$ will be denoted $F(r_t,y|r_1,\ldots,r_{t-1})$.

A popular instantiation of such a mechanism is the \emph{cost-function-based market maker}, in which $\R=\reals^k$ and $F(r_t,y|r_1,\ldots,r_{t-1}) = (r_t - r_{t-1})\cdot \phi(y) - (C(r_t) - C(r_{t-1}))$, where $C:\reals^k\to\reals$ is convex and $\phi:\Y\to\reals^k$~\citep{abernethy2013efficient}.
This payoff function corresponds to a trader making a fixed payment $C(r_t) - C(r_{t-1})$ in return for a bundle $r = r_t - r_{t-1} \in\reals^k$ of \emph{shares} in the securities $\phi_1,\ldots,\phi_k\in\reals^\Y$, i.e.\ $r_i$ units of security $\phi_i$, for an up-front cost paid to the market maker in terms of $C$.
Among the several nice properties of this market maker, one can see that a trader who believes $\E\,\phi = x$ has an incentive to buy or sell securities until $\nabla C(r_t) = x$, thereby revealing their belief.

\subsection{Scoring rule markets for properties}
\label{sec:srm-properties}

The goal of the market is to incentivize a good prediction for some property or statistic of $Y$, such as the median or mean.
Thus, much work considers prediction markets relying on \emph{proper scoring rules}, which are contracts designed to elicit a single agent's belief (i.e. the case $T=1$ of a market)~\cite{brier1950verification,gneiting2007strictly}.  While originally designed to elicit an entire distribution over the outcome $Y$, in many cases, for example when $\Y$ is very large or even infinite, one may be interested in obtaining only summary information about this distribution.  It is therefore natural to consider scoring rules which elicit such statistics, or \emph{properties}, of distributions.

Here and throughout the paper, $\P$ is a set of distributions of interest on the domain $\Y$, for example, the distributions with full support, with finite expectation, or so on.

\begin{definition}
  A \emph{property} is a function $\Gamma : \P \toto \R$, which associates a set of correct report values to each distribution.  We require $\Gamma$ to be \emph{non-redundant}, meaning $\Gamma^{-1}(r) \not\subseteq \Gamma^{-1}(r')$ for any reports $r,r'\in\R$ (i.e. we cannot have $r\in\Gamma(p) \implies r'\in\Gamma(p)$ for all $p$).
  A property is \emph{single-valued} if each $p$ maps to exactly one report, in which case we write $\Gamma$ as a function $\Gamma : \P \to \R$.
\end{definition}

A scoring rule $S(r,y)$ simply provides a payoff based on a reported value of the property and the observed outcome $y$.  We say the scoring rule \emph{elicits} the property if the correct report is incentivized.

\begin{definition}
  \label{def:elic}
  A scoring rule $S:\R\times\Y\to\reals$ \emph{elicits} a property $\Gamma:\P\toto \R$ if for all $p\in\P$,
  $\Gamma(p) = \argmax_{r\in\R} \E_p S(r,Y)$
  where $Y\sim p$.  When $\Gamma$ is single-valued, this condition becomes $\{\Gamma(p)\} = \argmax \left[\cdots\right]$.
  A property is \emph{elicitable} if some scoring rule elicits it.
\end{definition}

Not all properties are elicitable, meaning there is no way to score reports of their value based on an observed outcome in a manner which incentivizes truthfulness.  A classic example is the variance of a distribution, which does not have convex level sets (mixtures of distributions with the same variance have higher variance in general), a necessary condition for elicitability~\cite{lambert2008eliciting}.  (For such non-elicitable statistics, one could still discuss their elicitation \emph{complexity}, the number of reports needed to compute the desired property post-hoc; we discuss this concept in \S~\ref{sec:discussion}.)
However, several well-known statistics are elicitable properties, including expected values and means, medians, quantiles, expectiles, and ratios of expectations.

Combining the concept of scoring rules for properties with scoring rule markets yields the following natural prediction market mechanism for an arbitrary property $\Gamma:\P\toto\reals$ elicited by $S$~\cite{hanson2003combinatorial,lambert2008eliciting,abernethy2011collaborative}.
Initialize the market state at some $r_0 \in \R$.  When trader $t=1,\ldots,T$ arrives, they can choose to update the market state to any $r_t\in\R$, and once the outcome $y\in\Y$ is revealed, the market maker will pay the trader $S(r_t,y) - S(r_{t-1},y)$.
We can again express this mechanism using our $F$ notation above (we will later relate $F$ to $\Gamma$).

\begin{definition}
  \label{def:srm}
  A \emph{scoring rule market} for scoring rule $S:\R\times\Y\to\reals$ and initial state $r_0 \in \R$ is the mechanism $F(r_t,y|r_1,\ldots,r_{t-1}) = S(r_t,y) - S(r_{t-1},y)$.
\end{definition}

For brevity, we will simply write $F(r',y|r) = S(r',y) - S(r,y)$, or using contract notation $\vec F(r'|r) = \vec S(r') - \vec S(r)$, where of course $\vec S(r)_y = S(r,y)$.

\section{Axioms and Preliminaries}
\label{sec:axioms}
To motivate our choice of SRMs, we consider two preliminary axioms in Appendix \ref{app:derive-srms} that turn out to characterize SRMs; here, we briefly summarize.
\emph{Incentive-compatibility (IC)} states that agents maximize expected payment by choosing contracts that reveal their true belief about $\Gamma(p)$, for the property $\Gamma$ chosen by the market maker.
(Without IC, it is not clear in what sense a market reveals useful information.)
\emph{Path-independence (PI)}, an axiom appearing in prior work on prediction markets~\citep{abernethy2013efficient}, ensures that a participant cannot gain more by making multiple trades in a row than by simply making the single optimal trade immediately.
In Appendix \ref{app:derive-srms}, we show that any mechanism which satisfies both PI and IC is an SRM for a property $\Gamma$.
It is conceptually similar to results for markets eliciting the mean~\citep{abernethy2013efficient,abernethy2014general}, but more general as it holds for any property.
\begin{theorem} \label{thm:char-srm}
  A mechanism satisfies PI and IC for property $\Gamma$ if and only if it is a scoring rule based market (SRM) with some scoring rule $S$ that elicits $\Gamma$.
\end{theorem}
\noindent
Henceforth, unless otherwise noted we will assume PI and IC, which is to say we focus on SRMs.

\subsection{Arbitrage and Bounded Loss}
\label{sec:arb-loss}

We now express two time-honored axioms in our notation: no arbitrage, and bounded worst-case loss.  Recall that we write $\ones\in\reals^\Y$ to mean the contract paying out $\ones(y) = 1$ for each $y\in\Y$.  We also will write $\inf d = \inf_{y\in\Y} d(y)$ and $\sup d = \sup_{y\in\Y} d(y)$ to denote the payout bounds of contract $d\in\reals^\Y$; note that we use infima and suprema as $\Y$ may be infinite, for example when $\Y=\reals$.

The following axiom has been a desiderata since the inception of prediction markets: the market maker should not risk losing an unbounded amount of money.
\begin{axiom}[Worst-Case Loss (WCL)] \label{ax:wcl}
  An SRM $F$, initialized at $r_0$, satisfies WCL with bound $B \geq 0$ if for all $r\in\R$, $\sup \vec F(r|r_0) \leq B$.
\end{axiom}
\noindent
As of course $\vec F(r|r_0) = \vec S(r) - \vec S(r_0)$, WCL simply says that the scoring rule $S$ is bounded relative to the initial score $S(r_0,\cdot)$.

Another classic condition for a market mechanism is that a trader should never be able to profit regardless of the outcome $Y$.  In our contract notation, traders should never be able to purchase a contract $d$ for less than its minimum payoff $\inf d$, or equivalently, the net contract provided to traders should not be unconditionally positive.
\begin{axiom}[No Arbitrage (ARB)] \label{ax:arb}
  SRM $F$ satisfies ARB if $\forall\;r,r'\in\R$, $\inf \vec F(r'|r) \leq 0$.
\end{axiom}

\noindent
The no-arbitrage condition was crucial in deriving cost-function-based markets in terms of convex conjugate duality~\cite{abernethy2013efficient}. Surprisingly, it turns out that any SRM satisfies no-arbitrage automatically if it is incentive-compatible.
\begin{proposition}
  \label{prop:arb}
  SRM $F$ satisfies ARB if it is IC for some property $\Gamma\toto\R$.
\end{proposition}
\begin{proof}
  If $0 < \inf \vec F(r'|r) = \inf [\vec S(r') - \vec S(r)]$, then for all $y\in\Y$, $S(r',y) > S(r,y)$.  Thus, for any $p\in\P$, $\E_p S(r',Y) > \E_p S(r,Y)$, so $r$ cannot be an optimal report for any $p$, contradicting non-redundancy (as $\Gamma^{-1}(r) = \emptyset \subseteq \Gamma^{-1}(r')$).
\end{proof}

\subsection{New Axioms}
\label{sec:new-axioms}

We now identify several new axioms that capture desirable characteristics of prediction markets.
The first is motivated by traditional markets, wherein a trader can always ``neutralize'' or ``liquidate'' their holdings.
For example, if a trader buys a bar of gold, apart from apocalyptic scenarios, she can always sell it at any time for a strictly positive price.
As another example, a trader holding an Arrow-Debreu contract, paying \$1 upon event $E$ and \$0 otherwise, should be able to sell it for some nonzero (perhaps very low) price.
More generally, if the contract purchased is $d\in\reals^\Y$ and is non-constant, the trader should receive strictly more than $\inf d$ in cash (in both examples above, $\inf d = 0$).
This condition is the \emph{trade neutralization (TN)} axiom, which we now give along with two variants, one stronger and one weaker.
\begin{axiom}
  \label{ax:tn}
  An SRM satisfies \emph{Trade Neutralization (TN)} if for all trades $r_1\to r_1'$, and all reports $r_2$, there is a trade $r_2\to r_2'$ such that $\vec F(r_1'|r_1) + \vec F(r_2'|r_2) = c \ones$ for some scalar $c > \inf \vec F(r_1'|r_1)$.
\end{axiom}
\begin{axiom}
  \label{ax:pn}
  An SRM satisfies \emph{Portfolio Neutralization (PN)} if for all sets of trades $r_i\to r_i'$, $1\leq i\leq m$, and all reports $r$, there is a trade $r\to r'$ such that $\vec F(r'|r) + \sum_i \vec F(r_i'|r_i) = c\ones$ for some scalar $c > \inf \left[\sum_{i=1}^m \vec F(r_i'|r_i)\right]$.
\end{axiom}
\begin{axiom}
  \label{ax:wn}
  An SRM satisfies \emph{Weak Neutralization (WN)} if for all trades $r_1\to r_1'$, and all reports $r_2$, there is a trade $r_2\to r_2'$ such that $\inf \left[ F(r_1'|r_1)+F(r_2'|r_2)\right] > \inf F(r_1'|r_1)$.
\end{axiom}
Portfolio neutralization (PN) is stronger than TN: for any \emph{set} of purchased contracts, a trader should be able to sell the entire portfolio for at least its minimum payoff, all in one go, at any time.
Finally, weak neutralization (WN) asks just that she be able to ``make progress'' by making some trade that improves her worst-case payoff.  More precisely, WN states that traders holding a non-constant contract $d$ should be able to purchase a contract $d'$ such that the minimum payoff of their portfolio strictly increases, i.e., $\inf [d+d'] > \inf d$.
Note that WN is not far off from TN in the sense that a market maker allowing traders to exchange a contract $d$ for $(\inf d)\ones$, i.e., ``cash in'' $d$ for its worst-case payout, then the mechanism would effectively satisfy TN.
Naturally, we have PN $\implies$ TN $\implies$ WN;
surprisingly, we will see that in fact TN $\iff$ PN (Theorem~\ref{thm:tn-ic-implies-cost-func}).

Note that the inequalities above are all strict.
If we replaced them by weak inequalities, then WN would be trivially satisfied by
essentially every conceivable mechanism: by just keeping your current contract $d$, you are guaranteed a payoff of at least $\inf d$, by definition.
This makes intuitive sense, as unless the contract $d$ is constant, it strictly dominates $(\inf d)\ones$, so no rational trader would ever consent to trading $d$ for $(\inf d)\ones$.
The strict inequality thus captures a reasonable possibility for traders to ``sell back'' their previously-purchased contracts, in the sense that there is some belief a trader could hold where such a trade would be beneficial.

Our final axiom says that traders with (arbitrarily) small budgets can still expect to make money by participating in the mechanism.
\begin{axiom} \label{ax:sf}
  A market $F$ satisfies \emph{bounded trader budget (BTB)} if, for all market states $r$ and beliefs $p$ with $\Gamma(p) \neq r$, and for all $\epsilon > 0$, there exists a contract $F(r'|r)$ with $\inf F(r'|r) > - \epsilon$ and $\E_p F(r',Y|r) > 0$. 
\end{axiom}
One can equivalently interpret BTB as stating that, as the overall market scales larger relative to the budget (or maximum allowable loss) $\epsilon$ of individual traders, they still have incentives to participate.

\section{Central Examples}
\label{sec:first-examples}

With our axioms in hand, we now turn to specific properties: the mode, median, quantiles, and expectations.  For each, we wish to understand which axioms the corresponding SRMs can satisfy.  As we will see, each of these properties has a unique signature with respect to our axioms.

\subsection{Mode and Finite Property Markets}
\label{sec:mode}
Perhaps the simplest example of any SRM is the canonical mode market for a distribution on $k$ outcomes $\Y = \{1,\ldots,k\}$, with report space $\R=\Y$ and $S(r,y) = \ones\{y=r\}$ (``\$1 iff you guess correctly'').
Here we find that even the weakest version of neutralization, WN, is violated: a trade moving $r$ from $1$ to $2$ yields payoff $\ones\{y=2\}-\ones\{y=1\}$, meaning this contract will lose the owner $1$ if $Y=1$.
But if the current market state is, say, $3$, then no report the agent can make will avoid this loss of $1$ when $Y=1$.
For any report $r'$, the final payoff is $\ones\{y=2\}-\ones\{y=1\}-\ones\{y=3\}+\ones\{y=r'\}$.  See Figure \ref{fig:mode-position}.

\begin{figure}[t]
\begin{center}
  \includegraphics{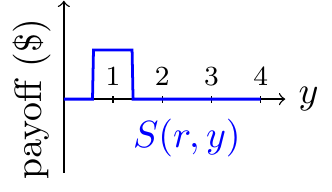}
  %
  \qquad
  \includegraphics{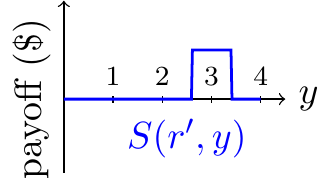}
  %
  \qquad
  \includegraphics{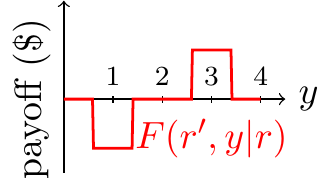}
  \end{center}
  \caption{\textbf{Visualizing trader ``position'' in a mode market.}
    Each curve gives the payoff of a given scoring rule or contract as a function of the outcome $y$.
    Here $S(r,y) = 1 \iff r=y$.
    Left: A trader who has reported $r=1$ stands to gain $1$ if $Y=1$ and gain $0$ otherwise.
    Center: Similarly for $r=3$.
    Right: A trader who moves the market from $r$ to $r'$ gets the function $F(r',\cdot|r) = S(r',\cdot) - S(r, \cdot)$.
    She stands to gain $1$ if $Y=r'$, lose $1$ if $Y=r$, and gain $0$ otherwise.
    } 
    \label{fig:mode-position}
\end{figure}
Intuitively, the lack of WN has negative implications for agents, who must take on risk that cannot be mitigated later.
It also violates BTB: The potential agent loss from any trade is a constant $1$, so an agent with a budget smaller than $1$ will not be able to participate.
This causes significant problems in practice for market designers as well, because the only possible solution, scaling down the scoring rule, is also unpleasant: many agents will be able to participate without much risk, so the market prediction will flip from outcome to outcome without necessarily aggregating information.

We emphasize that these characteristics are inherent to the mode, in the sense that they hold for any other scoring rule eliciting it.
In fact, the same conclusions hold for markets eliciting any \emph{finite property} $\Gamma:\P\to\{1,\ldots,k\}$, i.e. a property with a finite set of possible values.\footnote{See~\citet{lambert2009eliciting} for motivation and examples for the finite case.}

\begin{theorem} \label{thm:finite-markets}
  Any market for a finite property satisfies WCL, but not TN or BTB.
\end{theorem}

\subsection{Mean and expectation markets} \label{sec:expectation}

Consider now the mean of a random variable over the reals, $\R = \Y = \reals$.\footnote{In this setting, we naturally take $\P$ to be the set of beliefs with well-defined expectations.}
As an illustrative example, consider the scoring rule $S(r,y) = 2ry-r^2$, which elicits the mean.
The corresponding SRM takes the form $F(r',y|r) = r^2 - (r')^2 + 2y(r'-r)$. Perhaps the first observation one makes is that, because $y\in\reals$, any nontrivial trade leaves the trader exposed to unbounded potential loss, as well as unbounded potential gain.
This implies that BTB and WCL cannot hold.

What about trade neutralization?
Consider a trader holding the contract $\vec{F}(r_1',y|r_1) = \alpha_1 + \alpha_2 y$ for constants $\alpha_1 = (r_1)^2 - (r_1')^2$ and $\alpha_2 = 2(r_1'-r_1)$.
She would like to neutralize this position, so she must purchase some other contract of the form $\alpha_3 - \alpha_2y$, so that her net position will be the constant $\alpha_1 + \alpha_2$.
If the current market state is $r_2$, our hero can neutralize her previous trade by choosing $r_2'=(r_1-r_1')$, so that $2y(r_2' - r_2) = 2y(r_1' - r_1) = 0$.
Her worst-case liability decreases from $-\infty$ to a constant, showing that both WN and TN are satisfied.
And in fact, even if she holds a set of contracts, her net position is simply the sum and can still be written in the form $\alpha_1 + \alpha_2 y$ and the same argument goes through; this shows that PN is also satisfied.

\begin{figure}
\centering
\includegraphics{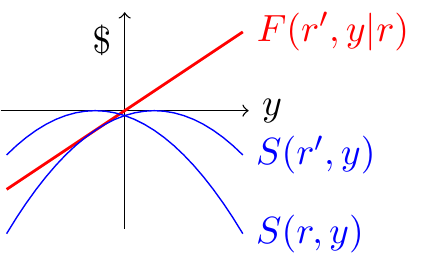}
%
\qquad
\includegraphics{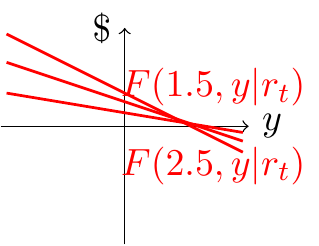}
%
\qquad
\includegraphics{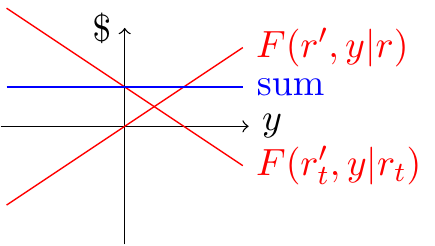}
%
  \caption{\textbf{Trader position in a mean market.}
    $S(r,y) = -(r-y)^2$.
    Left: Moving the market from $r=-1$ to $r'=1$ gives a position $F(r',\cdot|r)$ that pays off linearly in $y$.
    Center: At $r_t$, the trader chooses from a set of possible contracts, depending on if she reports $1.5$, $2.5$, etc.
    Right: choosing the contract $r_t'$ that neutralizes the previous trade $r \to r'$.
    }  
  \label{fig:mean-position}
\end{figure}

This example raises the questions of whether other markets eliciting the mean (if any) would have similar properties, and how this might depend on the random variable in question.

Motivated by this example, we now consider a much more general setting.  Let $\phi:\Y\to\reals^k$ be a given ``random variable'' and let $\Gamma(p) = \E_p[\phi(Y)]$ be its expected value.
Such a $\Gamma$ is called a \emph{linear property}.
We call SRM $S$ an \emph{expectation market} if $S$ is IC for such a $\Gamma$.
For ease of exposition (i.e.~to avoid relative interiors) we will assume that the range $\R=\Gamma(\P)$ is full-dimensional in $\reals^k$. 

Capitalizing on the nice characterization of scoring rules for expectations due to \cite{abernethy2012characterization,frongillo2015vector-valued} (Theorem~\ref{thm:mean-sr-char}), we know that any scoring rule $S$ eliciting $\Gamma$ takes the form
\begin{equation}
  \label{eq:sr-expectation}
  S(r,y) = G(r)+dG_r\cdot(\phi(y)-r) + f(y)
\end{equation}
for $G$ strictly convex with subgradients $\{dG_r\}_{r\in\R}$ and $f$ an arbitrary $\P$-integrable function.  Note however that the $f$ term will vanish in the definition of $F(r',y|r)$, so without loss of generality we may take $f(y)=0$ for all $y$.  We thus refer to expectation markets as being \emph{defined by $G$} if they satisfy Equation \ref{eq:sr-expectation} for that $G$ (for any $f$).
Letting $\conv(\phi(\Y))$ be the convex hull of the set of outcomes of the random variable $\phi$, we have the following characterization of worst-case loss for expectation markets:

\begin{proposition} \label{prop:msr-expectation-axioms}
  The linear property $\Gamma(p) = \E_p \phi(Y)$ has an expectation market satisfying WCL if and only if its domain $\conv(\phi(\Y))$ is bounded, in which case the market defined by any bounded $G$ satisfies WCL.
\end{proposition}

\begin{proposition} \label{prop:msr-expectation-btb}
  The linear property $\Gamma(p) = \E_p \phi(Y)$ has an expectation market satisfying BTB if its domain $\conv(\phi(\Y))$ is bounded, in which case the market defined by any differentiable $G$ satisfies BTB.
\end{proposition}

To illustrate, let us revisit the example at the beginning of this subsection, i.e. $\phi(y) = y$ and $S(r,x) = 2ry - r^2$.
(This corresponds to the convex function $G(r) = r^2$, up to a shift.)
We saw that this market satisfies neither BTB nor WCL for $\Y = \reals$, the entire real line.
However, for $\Y = (0,1)$, it satisfies both BTB and WCL, as $\conv(\phi(\Y)) = (0,1)$ and $G$ is bounded and differentiable on this domain.
Finally, recall that it also (intuitively) satisfied TN, and furthermore, PN.
This property can be shown to generalize, in particular, prior work has shown that linear properties have nice \emph{cost function} based markets.
Such markets take the form $F(r',y\mid r) = C(r) - C(r') + (r-r') \cdot \phi(y)$, for some convex function $C$ and $\R = \reals^k$.
It then follows relatively directly that such markets satisfy PN, as any position (set of contracts held by a trader) can be written in the form $d(y) = \alpha_1 + \alpha_2 \cdot \phi(y)$.
This can be interpreted as a fixed payment of $\alpha$ and $\alpha_{2,i}$ ``shares'' in a ``security'' $\phi(Y)_i$.
To neutralize, when the current market state is $r'$, it turns out to suffice to select the $r$ satisfying $r-r' = -\alpha_2$.
\begin{theorem}  \label{thm:linear-finiteY-pn}
  On a finite outcome space, i.e. $|\Y| < \infty$, for any linear property there exists an expectation market satisfying PN, WCL, and BTB.
\end{theorem}

\subsection{Median and Quantile Markets}
\label{sec:median}
We previously gave an example of a ``mean market'' given by $S(r,y) = 2ry-r^2$, the scoring rule analog of squared loss $L(r,y)=(r-y)^2$.
Perhaps the most natural statistic to investigate next is the median, elicited by the analog of ``absolute loss'', $S(r,y)=-|r-y|$.
What we find is surprising: unlike squared loss, the absolute loss market does not satisfy PN, and in fact does not even satisfy WN.
We show a general version of this result: no market for the median, or indeed any quantile, can satisfy WN.
On the other hand, while the mean market could not satisfy WCL except on bounded domains, there are median and quantile markets that can.

Our setting in this section is as follows.  Letting $\R=\outcomes=\reals$ and $\alpha \in (0,1)$, the \emph{$\alpha$-quantile} of probability distribution $p$ with continuous CDF\footnote{This assumption is dropped in Appendix \ref{app:quantile}.} is the $q_{\alpha}$ satisfying $\Pr_p[Y\leq q_\alpha(p)] = \alpha$.
Of course, the median is simply $q_\alpha$ for $\alpha=1/2$.

\begin{figure}[t]
\centering
\includegraphics{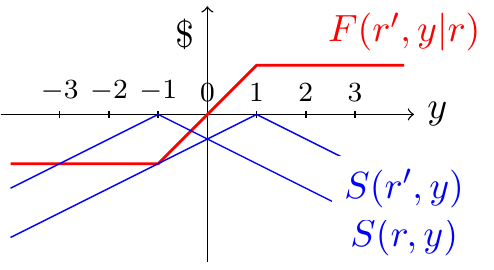}
%
%
\hspace*{48pt}
\includegraphics{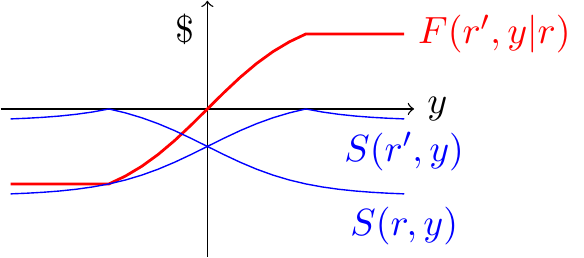}
%
%
\\[20pt]
\includegraphics{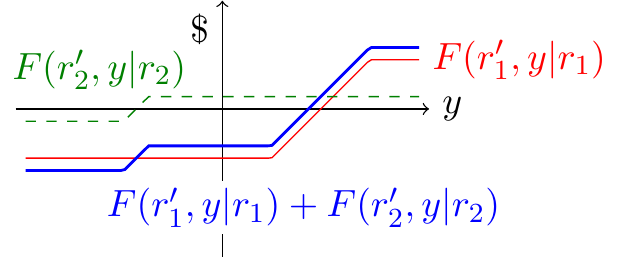}
%
\hspace*{20pt}
\includegraphics{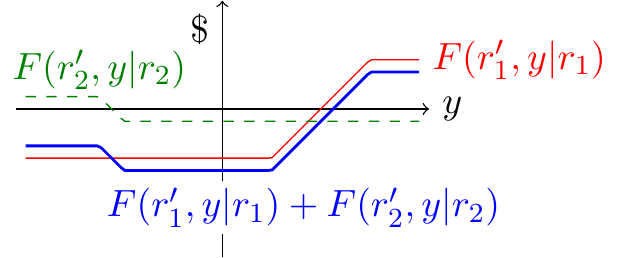}
%

\caption{\textbf{Trader position in a median market.}
  Here $S(r,y) = -|r-y|$.
  Top left: The payoffs for $S(r,\cdot)$ and $S(r', \cdot)$, and the contract that moves the market $r\to r'$, namely $F(r',\cdot|r) = S(r',\cdot) - S(r,\cdot)$.
  Top right: The same example but with $S(r,y) = -|g(r)-g(y)|$ with $g$ the sigmoid function.
  Bottom: Two examples where a trader with position ``red'' $F(r_1',y|r_1)$ considers a potential contract ``green'' $F(r_2',y|r_2)$.
  If purchased, the net position will be the blue curve.
  It sometimes falls below the original position, meaning the trader's worst case has gotten worse.} 
  \label{fig:median-position}
\end{figure}

We first show that quantile markets do not satisfy WN: there may not be trades which improve the trader liability at all.
Figure~\ref{fig:median-position} gives an example with absolute loss.
\begin{theorem} \label{thm:quantile-tn}
  No SRM for any $\alpha$-quantile satisfies WN.
\end{theorem}

Despite this negative result, quantile markets satisfy a surprisingly strong positive property.
Recall that the squared loss market with $S(r,y) = 2ry - r^2$, which elicits the mean of $Y$, could not hope to satisfy bounded worst-case loss if $\outcomes = \reals$.
And indeed, the absolute loss market $S(r,y) = -|r-y|$ shares this issue.
There is an elegant work-around, however: use the sigmoid function $g(r) = e^r/(1+e^r)$, or another strictly monotone transformation, to map reports continuously into the interval $(0,1)$.
Then $S(r,y) = -|g(r) - g(y)|$ is still proper, as strictly monotone functions commute with the median: for any $y$, all $y' \leq y$ are mapped below $y$ and all $y' \geq y$ are mapped above $y$, so the quantiles are simply mapped as well.
\begin{theorem} \label{thm:quantile-wcl}
  For all $\alpha\in(0,1)$, there is an SRM for the $\alpha$-quantile satisfying WCL.
\end{theorem}

Finally, quantiles also behave nicely with respect to bounded-budget traders.
To see this for $S(r,y) = -|r-y|$, notice that a small trade $r \to r+\epsilon$ carries a liability of only $\epsilon$, in the case $Y < r$.
(See Figure \ref{fig:median-position}.)
But any trade smaller than $\epsilon$ can carry positive expected payoff for a trader if it moves the market closer to her belief.
\begin{theorem} \label{thm:quantile-btb}
  If distributions in $\P$ do not contain point masses, then any $\alpha$-quantile market satisfies BTB.
\end{theorem}

\section{Characterizing Trade Neutralization}
\label{sec:char-tn}

In this section, we will characterize mechanisms that satisfy trade neutralization (TN) when the set of outcomes is finite, $|\Y| = n < \infty$.
Recall that TN captures the quality of traditional markets, particularly in the sense of a ``market maker'': the mechanism is always willing to buy back a contract it has previously sold, for a reasonable price.

TN may seem intuitively easy to satisfy, or at least might not appear to impose much structure on the market maker.
In fact, we will see that the opposite is the case.
To gain intuition for why TN might impose structure, recall that trades made by a participant must be interpretable as beliefs about the property the market is predicting.
This applies to both the purchase of a contract and its sale back to the mechanism.
Furthermore, this ``canceling'' trade must be available regardless of the current state of the market.
We will leverage this intuition to show that the mechanisms satisfying TN are quite special and closely related to expectation markets.

To characterize the SRMs that satisfy TN, it will prove useful to separate out ``contracts'' $d\in\reals^\Y$ from ``cash'' $\ones\in\reals^\Y$, even though technically both reside in the same space.  In particular, we would like a canonical way to take a contract $d$ and separate out its ``cash'' component as $d = d_0 + c\ones$, where intuitively $\$c$ is always paid regardless of the outcome, but $d_0$ depends on outcome.  To do this, we simply define $d_0$ as the projection of $d$ onto the hyperplane normal to $\ones$.

In general, for a scoring rule\footnote{Recall that given $\vec F$, one can define $\vec S(r) = \vec F(r|\emptyset)$. We will indicate throughout where we assume IC.} $S$, we define its range $\Srange = \{\vec{S}(r) : r \in \R\} \subseteq \reals^\Y$, and define the corresponding ``cashless'' contract space by $\Hc \defeq \Srange / \ones = \{\vec S(r) - (\vec S(r)\cdot\ones/|\outcomes|)\ones : r\in\R\}$.
Again, $\Hc$ is simply the projection of the range of $\vec S$ onto the hyperplane normal to $\ones$.
Similarly, we define $\Dc \defeq \Hc-\Hc = \{h_2-h_1 : h_1,h_2\in \Hc\}$ to be the set of possible score differences modulo $\ones$.
Equivalently, $\Dc$ is the projection of the range of $\vec F$ onto the same hyperplane.

We now show that if an SRM satisfies TN, its corresponding set $\Dc$ must have considerable structure: it must form an additive subgroup of $\reals^\Y$.
\begin{lemma} \label{lem:d-group}
  If SRM $F$ satisfies TN, then $\Dc$ is an additive subgroup of $\reals^\Y$ in the sense that $d,d' \in \Dc$ implies $-d, d+d' \in\Dc$.
  Moreover, the entire set $\Dc$ of contracts is available at all times: for any $h \in \Hc$, we have $\Dc = \Hc - \{h\}$.
\end{lemma}
\begin{proof}
  From the definition of TN, we must have the following: $\forall r_1,r_2,r_3\exists r_4$ with $F(r_2|r_1) + F(r_4|r_3) = c\ones$ for some $c\in\reals$.
  In terms of $\Hc$, as we have taken everything modulo $\ones$, this implies $\forall h_1,h_2,h_3\in \Hc \;\exists h_4\in \Hc$ such that $h_2-h_1+h_4-h_3 = 0$.  In other words, $h_1,h_2,h_3\in \Hc \implies h_3+h_1-h_2 \in \Hc$.

  To show that $\Dc$ is a group under vector addition, we must show closure under addition and the existence of additive inverses.  The latter is simple, as given any $d\in \Dc$ we can write $d = h_2-h_1$, so of course $-d = h_1-h_2 \in \Dc$.  We also have $0 = h_1 - h_1 \in \Dc$.  For closure, consider $d,d'\in \Dc$, and write $d=h_2-h_1$ and $d'=h_2'-h_1'$.  By the above, $(h_2+h_2'-h_1')\in \Hc$, so $h_2 - (h_1+h_1'-h_2') = d + d' \in \Dc$.

  Finally, we must show that for any $h\in h$, we have $\Dc = \Hc - \{h\}$.  Clearly $\Hc - \{h\} \subseteq \Dc$; for the converse, consider $d = h_1-h_2$.  By the above, $h + h_1 - h_2 \in \Hc$, and thus $d = (h + h_1 - h_2) - h \in \Hc - \{h\}$.
\end{proof}
For example, a construction allowed by Lemma \ref{lem:d-group} is $\Dc = \mathbb{Z}^k$ (the integer lattice), and this is indeed possible.
We give a $k=1$ example in Figure \ref{fig:discretized}.

We will need the following generalization of standard cost-function based markets:
\begin{definition}
  The \emph{generalized cost-function based market} parameterized by a convex ``cost function'' $C:\reals^k \to \reals$, ``securities'' $\phi:\Y\to\reals^k$, and ``share space'' $\Qc \subseteq \reals^k$ is the SRM of the form $F(r',y\mid r) = C(r) - C(r') + (r-r')\cdot\phi(y)$ with report space at market state $r$ given by $\R = \{r + q : q \in \Qc\}$.
\end{definition}
Recall that in standard cost-function markets, $\Qc = \reals^k$, so that any trade $d = r-r' \in \reals^k$ is allowable.
Furthermore, here $d_i$ is interpreted as the number of shares purchased of security $\phi_i$.
(For example, in the construction mentioned above, $\Dc = \mathbb{Z}^k$, we also have $\Qc = \mathbb{Z}^k$, so traders may only purchase integer numbers of shares in this example.)

\begin{definition}
  \label{def:open-market}
  A cost-function based market defined by $C$ and $\phi$ is \emph{open} if $C$ is differentiable and\footnote{We write $\mathrm{int}(A)$ for the interior of the set $A$ and $\conv(A)$ for its convex hull.} $\{\nabla C(q) : q \in \reals^k\} = \mathrm{int}(\conv(\phi(\Y)))$.
  It is \emph{quasi-open} if for every pair of ``share vectors'' $q,v\in\Qc$ and subgradient $x\in\partial C(q)$, we have $x\cdot v < \max_{y\in\Y} v\cdot \phi(y)$.
\end{definition}
To see that quasi-openness is a relaxation of openness, we comment that if we required the condition to hold for all $v \in \reals^k$, then it would be equivalent to $\cup_{q\in\Qc} \partial C(q) \subseteq \mathrm{int}(\conv(\phi(\Y)))$, which is a clear relaxation of openness; and quasi-openness requires somewhat less, as it only must hold for all $v \in \Qc$.

\begin{theorem} \label{thm:tn-ic-implies-cost-func}
  Let SRM $F$ which is IC for $\Gamma$ be given.  Then $F$ satisfies TN if and only if it is a generalized cost-function-based market for securities $\phi:\Y\to\reals^k$, where:
  \begin{enumerate}[leftmargin=0.8cm,label={(\roman*)}]
  \item The share space $\Qc$ of possible purchases is an additive subgroup of $\reals^k$, and
  \item The market is quasi-open.
  \end{enumerate}
  Moreover, TN implies PN.
\end{theorem}
\begin{proof}
  For the forward direction, that TN implies cost-function-based and PN, we proceed in five parts: (1) construct a basis for $\Dc$, which will be the securities $\phi$; (2) rewrite $\vec F$ in terms of that basis; (3) show that the map from reports to ``shares'' is bijective; (4) extract the cost function $C$ from the share representation of $F$; (5) show that $C$ satisfies our relaxed version of openness.  For the converse, we will appeal to Theorem~\ref{thm:cost-pn}.

  \emph{1. Basis of $\Dc$.}
  \quad
  By Lemma~\ref{lem:d-group}, $\Dc$ is an additive subgroup of $\reals^\Y$.  Let $d^1,\ldots,d^k \in \Dc$ be a basis for the linear span of $\Dc$.  We can now write our securities in terms of this basis: define $\phi:\Y\to\reals^k$ by $\phi(y)_i = d^i_y$.  It will be convenient to work with $\phi$ in matrix form $\Phi\in\reals^{n\times k}$, which naturally we define as $\Phi_{y,i} = \phi(y)_i$.  Thus, we now have $\Dc = \{\Phi \cdot q : q\in \Qc\}$ where $\Qc$ is an additive subgroup of $\reals^k$.

  \emph{2. Rewrite $\vec F$.}
  \quad
By the definition of $\Dc$ as the range of $F$ modulo $\ones$, and the decomposition above, we know that for every $r,r'\in\R$ we can write $\vec F(r'|r) = \Phi \cdot v + c\ones$ for some $v\in\reals^k$ and $c\in\reals$.  Thus, letting $r_0$ be the initial state of the given SRM, we have functions $g:\R\to\reals$ and $v:\R\to\reals^k$ such that for all $r\in\R$,
$\vec F(r|r_0) = \Phi \cdot v(r) + g(r) \ones$.  By definition $\Phi \cdot v(r) \in \Dc$ and thus we must have $v(r) \in \Qc$, meaning we can write $v : \R \to \Qc$.  Note that the expected payoff when $Y\sim p$ can now be written $\E_p F(r,Y|r_0) = p^\tr \vec F(r|r_0) = p^\tr \Phi\,v(r) + g(r) p^\tr \ones = \E_p \phi \cdot v(r) + g(r)$.  In particular, $\Gamma$ can only depend on $p$ through $\E_p\phi$, so letting $\Xc = \conv(\phi(\Y))$ as before, we have some $\psi:\Xc\toto\R$ such that $\Gamma(p) = \psi(\E_p \phi)$.

  \emph{3. The map $v : \R \to \Qc$ is a bijection.}
  \quad
  Surjectivity of $v$ follows from the fact that transformation by the change of basis matrix $\Phi$ is a bijection, as by definition we are simply rewriting elements of $\Dc$ via elements of $\Qc$.    For injectivity, we will use IC.  Suppose $v(r) = v(r')$.  If $g(r)=g(r')$, $\Gamma$ is not non-redundant as $\vec F(r|r_0) = \vec F(r'|r_0)$, so clearly $\Gamma(p) = r \iff \Gamma(p) = r'$.  Thus without loss of generality $g(r) > g(r')$, which implies $r' \notin \argmax_{x\in\R} v(x) \cdot \E_{p'} \phi(Y) + g(x) = \argmax_{x\in\R} \E_p F(x,Y|r_0)$ for any $p'\in\P$, so $r'$ cannot be in the range of $\Gamma$ at all (i.e.~the report $r$ dominates $r'$), a contradiction.  Thus, $v$ is injective and hence a bijection.

  \emph{4. Extracting the cost function.}
  \quad
  Now that we have a bijection from reports to ``shares'', intuitively, we should just be able to take $C(q) = -g(v^{-1}(q))$.  More care is needed, however, as $C$ must be convex with the correct subgradients.  Let us define a convex function $G:\Xc\to\reals$ by $G(x) = \sup_r v(r)\cdot x + g(r)$, which is the optimal expected score when $\E_p \phi = x$.  By IC and the definition of $\psi$ above, we have $G(x) = v(r)\cdot x + g(r)$ if and only if $r\in\psi(x)$.  Now we can define $C:\reals^k\to\reals$ as the conjugate of $G$, namely $C(q) = \sup_{x\in\Xc} q\cdot x - G(x)$.  Thus, we need only show that $C(v(r)) = -g(r)$, as alluded to above, as this would imply $F(r,y|r_0) = v(r)\cdot\phi(y) - C(v(r))$ as desired.

To make headway, we will appeal to convex analysis.  First, we will establish that $G(x) = v(r)\cdot x + g(r) \iff v(r) \in \partial G(x)$, meaning if $r$ is the optimal report for expected value $x$, then $v(r)$ is a subgradient of $G$.  This follows directly from the subgradient inequality; if $r$ is optimal for $x$, then for all $x'$, $G(x') = sup_{r'} v(r')\cdot x + g(r') \geq v(r)\cdot x + g(r) = G(x) + v(r)\cdot(x'-x)$, so $v(r)\in\partial G(x)$.  Conversely, suppose $v(r)\in\partial G(x)$ but $r$ is optimal for some $x'$, so $G(x') = v(r)\cdot x' + g(r)$.  The subgradient inequality gives $v(r)\cdot x' + g(r) = G(x') \geq G(x) + v(r)\cdot(x'-x)$, which implies $v(r)\cdot x + g(r) \geq G(x)$, so $r$ must be optimal for $x$ as well: $G(x) = v(r)\cdot x + g(r)$.

In summary, we have $r \in \psi(x) \iff G(x) = v(r)\cdot x + g(r) \iff v(r) \in \partial G(x)$.  By \cite[Thm E.1.4.1]{urruty2001fundamentals}, we immediately have $C(q) = q\cdot x - G(x) \iff q\in\partial G(x)$.
Putting these together, we have $r \in \psi(x) \iff v(r) \in \partial G(x) \iff C(v(r)) = v(r)\cdot x - G(x) = -g(r)$.  Thus, as every $r\in\R$ is in the image of $\psi$, we must have $C(v(r)) = -g(r)$ for all $r$.

  \emph{5. $C$ is (almost) open.}
  \quad
  We have established that $F$ can be written as a cost-function-based market with securities $\phi$ and a potentially nondifferentiable $C$.
To conclude this direction of the proof, suppose that for some $x\in\partial C(q)$ we have $x\cdot v \geq \max_{y\in\Y} v\cdot\phi(y)$.
Following the proof of Theorem~\ref{thm:cost-pn}, we start the trader at $q$ selling $v$ (as $(-v)\in\Qc$) and place the state back at $q$.
To neutralize, the trader must now purchase $v\in\Qc$, but by the same argument as before (using the fact that $\partial C(\reals^k) \subseteq \Xc$ by construction of $C$), $C(q+v)-C(q) = x\cdot v \geq \max_{y\in\Y} v\cdot\phi(y)$, which contradicts TN.

  \emph{Converse.}
  \quad
  Finally, for the reverse direction, we can simply repeat the proof of Theorem~\ref{thm:cost-pn}: if the trades $v_i$ are elements of the group $\Qc$ for all $v_i$, then so is $v = \sum_{i=1}^w v_i$.
It only remains to be shown that the price is less than the worst-case payoff of $v$, which follows from integrating the guarantee in (ii) as in Lemma~\ref{lem:cost-function-prices}.
\end{proof}

\begin{figure}
  \includegraphics[width=0.49\textwidth]{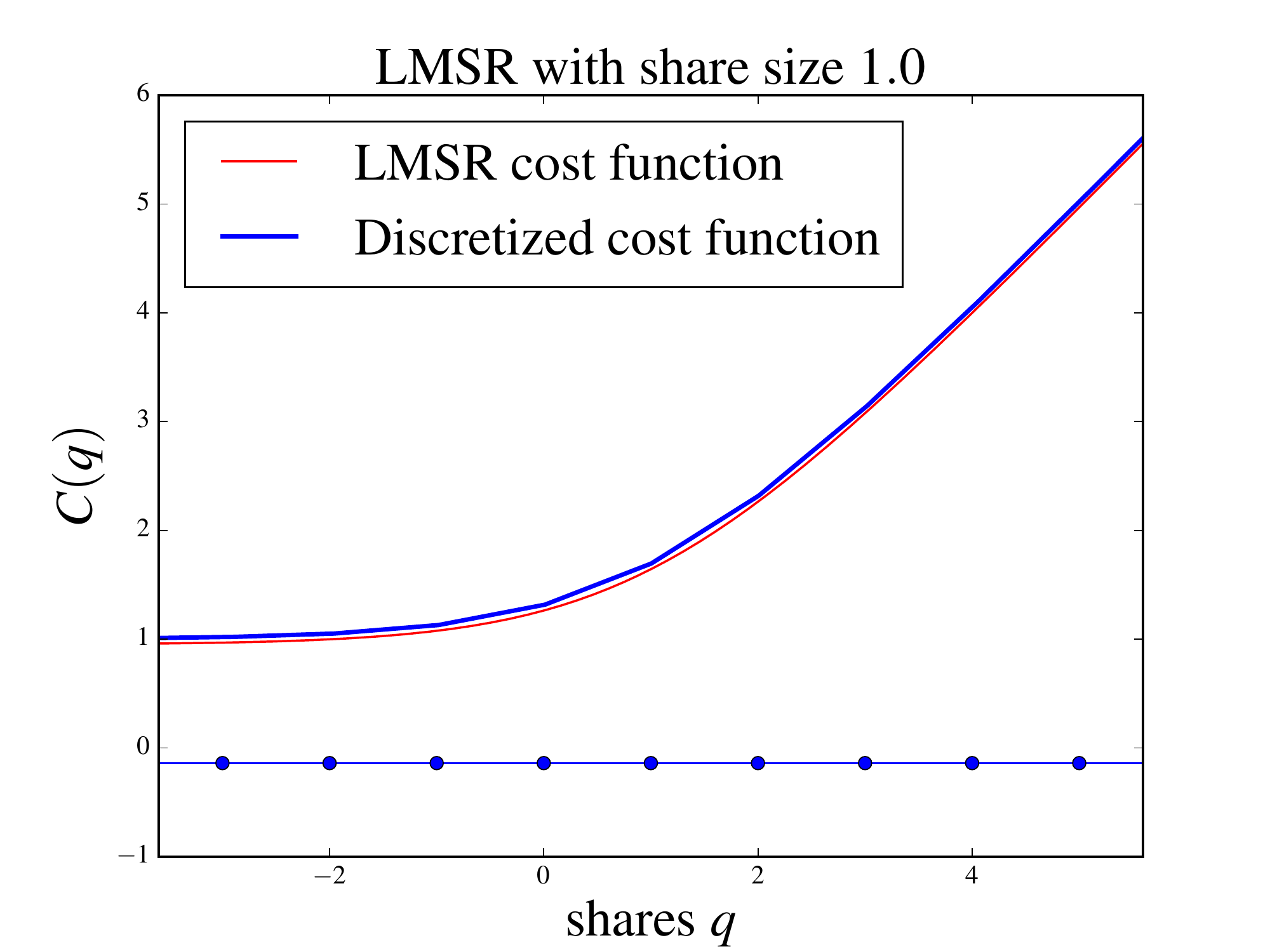}
  \hfill
  \includegraphics[width=0.49\textwidth]{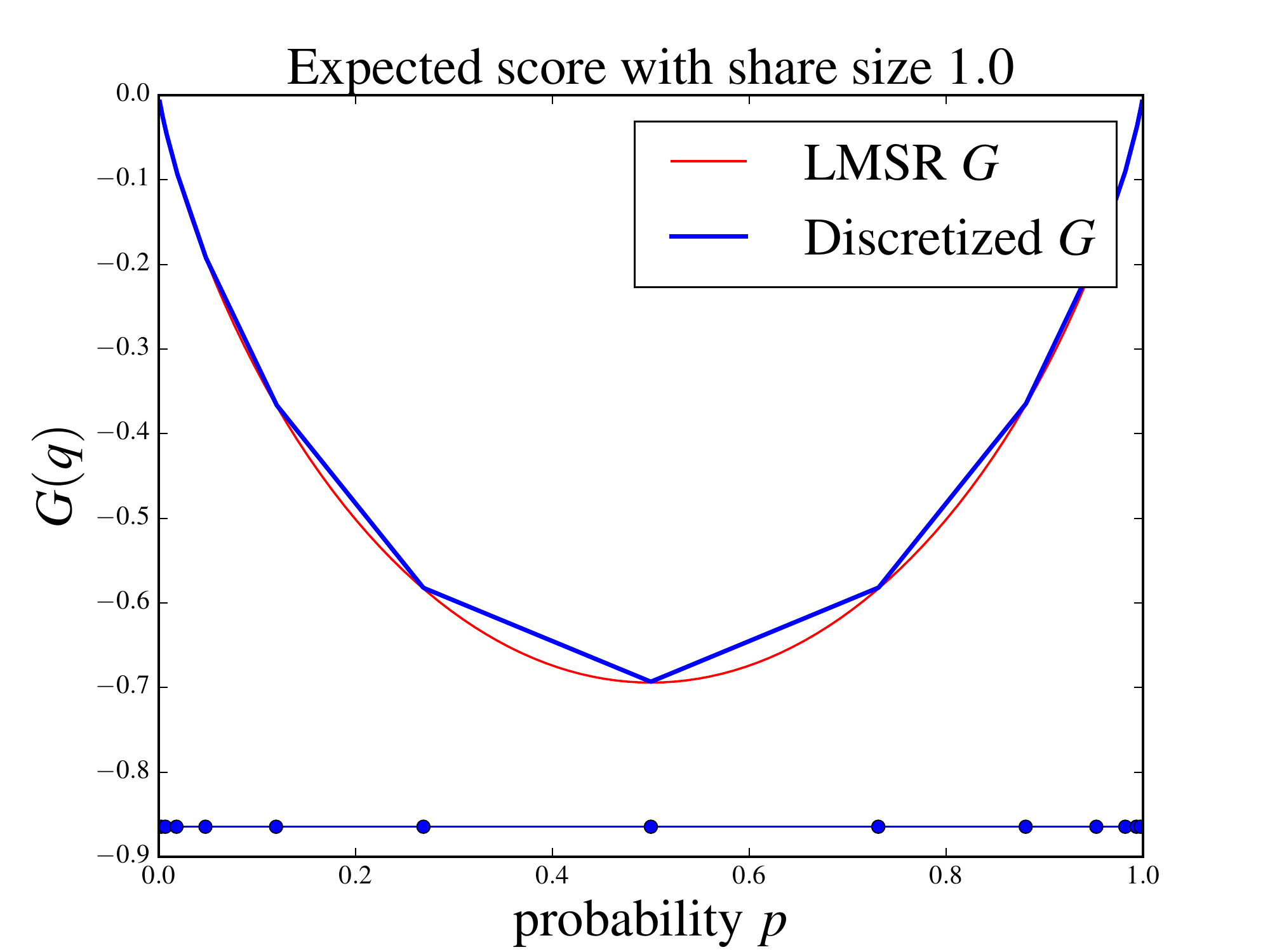}
  \caption{\textbf{The ``discretized'' LMSR.}
    This market for $\Pr[Y=1]$ allows only integer trades of a single security, using the LMSR cost function $C(q) = \log(1 + e^q)$ and its conjugate $G(p) = p \log(p) + (1-p)\log(1-p)$.
    Traders may move the market from any ``dot'' to any other.
    The possible market states are equally spaced in share space (horizontal axis, left figure), translating to \emph{non}-equally spaced \emph{prices} (horizontal axis, right figure).
    Theorem \ref{thm:tn-ic-implies-cost-func} implies that all mechanisms satisfying PI, IC, and TN have this general format, being discretizations of cost-function markets.
    } 
  \label{fig:discretized}
\end{figure}

This result shows that, despite the seeming weakness of the TN condition, the only SRMs satisfying TN are cost-function-based, or variations thereof (see Figure~\ref{fig:discretized}).
The result also raises the question of whether any market can satisfy WN, or if perhaps WN is equivalent to TN in a formal sense. We will conclude by mentioning some interesting additional results answering this question and discussing future directions.

\section{Share-Like Market for Ratios of Expectations} \label{sec:ratios}

From Theorem~\ref{thm:tn-ic-implies-cost-func}, we know that any SRM satisfying the TN axiom will effectively be a cost-function-based market dealing in shares of some set of securities, perhaps restricting trades to some discretization of the share space.
As a corollary, all such TN markets must elicit an expected value either directly or indirectly, or perhaps a discretization thereof.
Theorem~\ref{thm:tn-ic-implies-cost-func} leaves open the possibility, however, of a ``share-like'' market for a property other than an expected value which satisfies WN but not TN.
We now give one such example: a ratio of expectations.

Before diving into the formalism, let us expand on the intuition given in \S~\ref{sec:msr-introduction}.
In a cost-function-based market, traders purchase bundles $r\in\reals^k$ of securities $\phi:\Y\to\reals^k$, in exchange for paying some up-front cost $C(r_t+r) - C(r_t)$ to the market maker in the form of cash.
In other words, traders are exchanging the bundle $r$ of securities $\phi$ for an amount $C(r_t+r) - C(r_t)$ of the security $\ones$.
The ratio of expectations market can be thought of similarly, but now traders exchange the bundle $r$ of securities $\phi$ for an amount $C(r_t+r) - C(r_t)$ of some new  security $b$, which by convention will be nonnegative (and bounded away from $0$ for infinite $\Y$).
In other words, the new market is exactly the same as the old, but the cost function is in units of the security $b$ rather than cash ($\ones$).
Immediately we can glean that, just as in a cost-function-based market traders with belief $p$ had an incentive to trade until $\nabla C(r_t) = \E_p \phi$, i.e.\ the change in cost is equal to the expected security payoff, now traders will trade until $\nabla C(r_t) \E_p b = \E_p \phi$.
Thus, the prices of the new market will be the market's consensus about the ratio of expectations of $\phi$ and $b$, $\nabla C(r_t) = \E_p \phi / \E_p b$.

What axioms does this new market satisfy?
It would seem that its share-like nature could somehow circumvent Theorem~\ref{thm:tn-ic-implies-cost-func} and satisfy TN, but of course this is not the case.
As we will show, however, the ratio of expectations market does satisfy WN.
First, we must formally define our new market.

Let $\Y$ be a finite set.
For $\phi:\Y\to\reals^k$ and $b:\Y\to\reals$ with $\inf b > 0$, we define the ratio of expectations $\Gamma(p) = \E_p \phi / \E_p b$.
As usual we assume that $\phi$ is affinely independent, so $\{\E_p \phi : p\in\P\}$ is full-dimensional, and that $\inf b > 0$, which since $\Y$ is finite is equivalent to the payoffs $b(y)$ being nonnegative for all $y\in\Y$.
A characterization of scoring rules eliciting $\Gamma$ was shown by Frongillo and Kash~\citep{frongillo2015vector-valued}, which extended the real-valued case given by Gneiting~\citep{gneiting2011making}: a scoring rule $S$ elicits $\Gamma$ if and only if
\begin{equation}
  \label{eq:ratio}
  S(r,y) = b(y) G(r) + dG_r \cdot (\phi(y) - r b(y)) + f(y)~,
\end{equation}
where as in Theorem~\ref{thm:mean-sr-char}, $G:\R\to\reals$ is strictly convex with selection of subgradients $dG$, and $f$ is arbitrary (and, as usual, irrelevant for SRMs).

Our main result for this section is that essentially all ``open'' markets for $\Gamma$ satisfy WN.
The proof first applies convex conjugate duality to arrive at the cost-function-like market described above, and then applies facts about the usual cost-function-based framework (e.g. Lemma~\ref{lem:cost-function-prices}) to show that the prices for trading in a bundle must be strictly less than its worst-case payoff.

\begin{theorem}
  \label{thm:ratio-tn}
  A ratio-of-expectations market for differentiable $G$ and $\nabla G(\R)=\reals^k$ satisfies WN.
\end{theorem}
\begin{proof}
  We construct a ``cost function'' $C(q) = \sup_{r\in\R} r\cdot q - G(r)$ as the convex conjugate of $G$.
  The corresponding scoring rule is $S_C(q,y) = \phi(y)\cdot q - C(q)b(y)$.
  By assumption, $\nabla G:\R\to\reals^k$ is a bijection, and thus markets $F_G$ with $G$ and $F_C$ with $C$ are equivalent.  
  It therefore suffices to show that $F_C$ satisfies WN.

  Consider any trade $q_1\to q_1'$ and let $v = q_1'-q_1$.
  For any $q_2$, we will show that the trade $q_2' = q_2-v$ satisfies WN.
  The difference in net payoff between the first and second trade is $(C(q_2)-C(q_2-v))b(y) - v\cdot \phi(y)$, so it suffices to show that this is always positive.

  By the proof of Lemma~\ref{lem:cost-function-prices}, we then have for any $q,w\in\reals^k$ 
  \begin{equation}
    \label{eq:ratio-prices}
    C(q+w) - C(q) < \sup_{x \in \R} x\cdot w = \max_{y\in\Y} \; (\phi(y)/b(y))\cdot w~.
  \end{equation}
  Taking $q=q_2$ and $w=-v$, we have $C(q_2)-C(q_2-v) < \phi(y)\cdot (-v)/b(y)$ for all $y$, which is equivalent to $(C(q_2-v) - C(q_2))b(y) >  \phi(y)\cdot v$ for all $y$, thus establishing WN.
\end{proof}

By establishing WN for SRMs eliciting a ratio of expectations, we are essentially saying that these markets are ``reasonable'' and could plausibly aggregate trader beliefs effectively.
It would be interesting to see how such markets perform in practice.

\section{Discussion and Directions}
\label{sec:discussion}

We have presented market axioms for scoring rule markets (SRMs) and studied a handful of well-known statistics/properties to see which of these axioms their corresponding markets can satisfy.  Interestingly, we have seen wide variation among the satisfied axioms, most dramatic of which are the neutralization axioms: TN is satisfied essentially exclusively by cost-function-based markets, whereas median/quantile and mode markets do not even satisfy its weakest version WN.

Perhaps the first question that comes to mind in light of our results is, are there any non-cost-function-based markets satisfying WN?  We now give two positive examples: expectiles and ratios of expectations.  As discussed in \S~\ref{sec:axioms}, WN is essentially equivalent to TN, and thus these markets could be expected to perform similarly to cost-function-based markets.  We then conclude with remarks and future directions.

\paragraph{Other properties.}

We have studied the mean (expected values), median, mode, and ratio of expectations in this paper, but clearly many other properties of interest remain.
One interesting example is the $\tau$-expectile, a type of generalized quantile introduced by Newey and Powell~\citep{newey1987asymmetric}, is defined as a solution $x=\mu_\tau$ to the equation $\E_p\left[|\ones_{x\geq y}-\tau|(x-y)\right]=0$.
Scoring rules for expectiles take a hybrid form of those for means and quantiles, e.g.~the \emph{asymmetric squared error}, $S(r,y) = -|\ones\{y\leq r\} - \tau|(y-r)^2$.
We show in Appendix~\ref{app:expectiles} that any SRM for expectiles satisfies WN, as the score is eventually linear, and thus can be neutralized by simply matching the slopes (analogous to Figure~\ref{fig:mean-position}).

\paragraph{Other market forms.}

\citet{wolfers2004prediction} study prediction markets from a more empirical perspective, and suggest a number of possible market formulations.
One which is not covered here is to offer contracts $d^r$ such that $\E_p d^r = 0 \iff \Gamma(p) = r$.
(This $V(r,y) = d^r(y)$ is called an \emph{identification function}.)
They illustrate this idea with a market eliciting the median of $Y$, where $d^r$ pays $\$1$ if $Y>r$ and $-\$1$ otherwise.
While the properties of such a contract space would be interesting to study in a continuous double-auction setting, one may ask how to translate it to an automated market maker setting.
Here the market would maintain a centralized median $r_t$ and traders could either buy or sell $d^{r_t}$, moving $r_{t+1} \gets r_t \pm \epsilon$.
Unfortunately, for any $\epsilon>0$, such a market has unbounded worst-case loss even on a bounded domain (traders buy and sell between $r=0$ and $r=\epsilon$ and $y=\epsilon/2$).
For infinitesimal $\epsilon$ however, the market becomes the absolute-loss SRM with $S(r,y) = -|r-y|$, because $d^r = \tfrac {d}{dr} S(r,y)$.

\paragraph{Complexity.}

Finally, we remark that the recent concept of \emph{elicitation complexity} brings interesting implications for our study.
Here one asks, given a property $\Gamma$ of interest, how many dimensions $k$ does one need for there to exist an elicitable $\Gamma':\P\to\reals^k$ from some ``nice'' class of properties, where one can compute $\Gamma$ from $\Gamma'$ via a \emph{link function} $f$, i.e.~$\Gamma = f \circ \Gamma'$.
In our context, one may choose ``nice'' to mean a property having an SRM satisfying any one of the axioms we discuss.
The most natural may be TN, in which case one is essentially asking $\Gamma'$ to be an expected value, a case studied by~\citet{agarwal2015consistent}.
It would be interesting to characterize properties having markets satisfying WN, and identify properties with low complexity with respect to these WN properties.

\paragraph{Acknowledgments.}
Thanks to Yiling Chen for helpful comments and pointers.

\bibliographystyle{plainnat}
\bibliography{diss,refs}

\vfill
\break

\appendix

\section{Axioms} \label{app:derive-srms}

Our first axiom, incentive compatibility, is our formalization of the purpose behind designing the mechanism in the first place: it ensures that we can reasonably interpret the reports of the participants as predictions about the unknown outcome $y\in\Y$.

\begin{axiom}[Incentive Compatability (IC)] \label{ax:ic}
  Given a property $\Gamma : \P \toto \R$, a market $F$ satisfies IC for $\Gamma$ if for all histories $r_1,\ldots,r_t\in\R$, and all beliefs $p\in\P$, we have $\Gamma(p) = \argmax_{r\in\R} \E_p[F(r,Y|r_0,\ldots,r_t)]$.
\end{axiom}

\begin{axiom}[Path Independence (PI)] \label{ax:pi}
  A mechanism $F$ satisfies PI if for all $r,r'\in\R$ and $r_1,\ldots,r_t\in\R$, $\vec F(r|r_1,\ldots,r_t) = \vec F(r'|r_1,\ldots,r_t) + \vec F(r|r_1,\ldots,r_t,r')$.
\end{axiom}
PI states that, given a history $r_1,\ldots,r_t$, the following contracts are equivalent: reporting some $r'$, then coming back immediately and reporting $r$; and directly reporting $r$.
This enforces that agents cannot gain by making spurious prophecies and then immediately correcting them.

The following result says that any mechanism satisfying IC and PI must be an SRM.
The idea conceptually as well as the proof are very similar to that of \citet{abernethy2013efficient}.
The implications are more general, however, as the result applies to any mechanism that offers contracts at each point in time.
\begin{theorem*}[\ref{thm:char-srm}]
  A mechanism satisfies PI and IC for property $\Gamma$ if and only if it is a scoring rule based market (SRM) with some scoring rule $S$ that elicits $\Gamma$.
\end{theorem*}
\begin{proof}
  Let scoring rule $S$ eliciting $\Gamma$ be given.  By standard arguments~\cite{hanson2003combinatorial,lambert2008eliciting}, the corresponding SRM satisfies IC: the agent's expected payoff is $\E_p F(r_t,Y|r_1,\ldots,r_{t-1}) = \E_p S(r_t,Y) - \E_p S(r_{t-1},Y)$, which is maximized at $r\in\Gamma(p)$ as $S$ elicits $\Gamma$ and the second term is a constant that does not depend on her report.
  This SRM also satisfies PI by the telescoping nature of the payments: $\vec F(r'|r_1,\ldots,r_t) + \vec F(r|r_1,\ldots,r_t,r') = \vec S(r') - \vec S(r_t) + \vec S(r) - \vec S(r') = \vec S(r) - \vec S(r_t) = \vec F(r|r_1,\ldots,r_t)$.  (Recall that we defined $\vec S:\R\to\reals^\Y$ by $\vec S(r)_y = S(r,y)$.)

  For the converse, let $F$ be given, and define $\vec S(r) = \vec F(r|\emptyset)$, where $\emptyset$ denotes an empty.  Applying PI for $t=1$ we have, $\vec F(r | r_1) = \vec F(r | \emptyset) - \vec F(r_1 | \emptyset) = \vec S(r) - \vec S(r_1)$.
  Take this as the base case for an induction on $t$.
  For any $t > 1$, PI gives us
  \begin{align*}
    \vec F(r | r_1,\ldots,r_t) &= \vec F(r | r_1,\ldots r_{t-1}) - \vec F(r_t | r_1,\ldots, r_{t-1})  \\
                               &= \vec S(r) - \vec S(r_{t-1}) ~ - ~ \left[ \vec S(r_t) - \vec S(r_{t-1}) \right]  \;\;
                               = \;\; \vec S(r) - \vec S(r_t)~.
  \end{align*}
  Finally, by definition of $S$ and because $F$ satisfies IC for $\Gamma$, $S$ elicits $\Gamma$. 
\end{proof}

\section{Mode and Finite-Property Markets} \label{app:finite}
\begin{theorem*}[Theorem \ref{thm:finite-markets}]
  Any market for any finite property satisfies Axiom \ref{ax:wcl} (WCL), but not Axioms \ref{ax:tn} (TN) or \ref{ax:sf} (BTB).
\end{theorem*}
\begin{proof}
  There is a finite set of possible contracts $D = \{F(r|r') : r,r' \in \R\}$.
  This implies WCL, since by the telescoping property of payoffs, the mechanism's loss is some contract, hence bounded.
  It also implies violation of BTB, for the case where a trader's budget is less than $\inf_{d \in D} \inf d$.

  Lemma \ref{lem:d-group} shows that the set of contracts must be closed under addition in order to satisfy TN.
  In particular, they must be an infinite set (or consist only of the zero vector, which will not elicit anything).
  So finite property markets cannot satisfy TN.
\end{proof}

\section{Mean and Expectation Markets} \label{app:expectation}

We note the characterization of scoring rules for the expectation, which has been established with varying degrees of generality \cite{savage1971elicitation,abernethy2012characterization,frongillo2015vector-valued}:
\begin{theorem}
  \label{thm:mean-sr-char}
  A scoring rule $S$ elicits $\Gamma : p \mapsto \E_p[\phi(Y)]$ if and only if $S(r,y) = G(r)+dG_r\cdot(\phi(y)-r) + f(y)$ for some $G$ strictly convex with subgradients $\{dG_r\}_{r\in\R}$, and $f$ an arbitrary $\P$-measurable function.\footnote{Typically scores are allowed to take on values in $\reals\cup\{\infty\}$, essentially to accommodate the log scoring rule, but we will typically restrict to the relative interior of the domain anyhow, thus avoiding this issue.}
\end{theorem}

\begin{proposition*}[Proposition \ref{prop:msr-expectation-axioms}]
  The linear property $\Gamma(p) = \E_p \phi(Y)$ has an expectation market satisfying WCL if and only if its domain $\conv(\phi(\Y))$ is bounded, in which case the market defined by any bounded $G$ satisfies WCL.
\end{proposition*}
\begin{proof}
  Lemma \ref{lem:msr-g-bound-wcl} proves that on a bounded domain, WCL is equivalent to boundedness of $G$.
  (Some bounded convex $G$ always exists for a bounded domain.)
  Lemma \ref{lem:msr-domain-wcl-btb} proves that WCL cannot be satisfied on an unbounded domain.
\end{proof}

\begin{lemma} \label{lem:msr-g-bound-wcl}
  On a bounded domain $\conv(\phi(\Y))$, an expectation market defined by $G$ satisfies WCL if and only if $G$ is bounded.
\end{lemma}
\begin{proof}
  
  Let the initial market state be any $r$ where $d_G(r)$ is finite, such as any $r$ in the interior of the domain.\footnote{This technical condition only rules out boundary such as the $\log$ scoring rule with initial prediction $p=0$.}
  We show that for this fixed $r$, worst-case loss is bounded by a constant if and only if $G$ is bounded.

  Recall from the characterization that $S(r,y) = G(r) + d_G(r) \cdot \phi(y)$.
  By bounded domain, $\|\phi(y)\|$ is bounded, so there exists a constant $B$ with $-B \leq S(r,y) \leq B$ for all $y$.\footnote{If $f(y) \neq 0$, then we can say $S(r,y) + f(y) \in [-B,B] + f(y)$, and we would see $f(y)$ cancel out later in the proof.}

  The loss of the market maker when the final state is $r'$ is $S(r', y) - S(r, y)$, and the worst-case loss is $WCL \defeq \sup_{r',y} \left[ S(r',y) - S(r,y\right]$.
  We have $WCL \in \left(\sup_{r',y} S(r',y)\right) \pm B$, hence it is bounded if and only if the first term is.
  For each $y$, the supremum over $r$ is achieved at $r=\phi(y)$ by properness of the scoring rule ($r$ is the optimal report for the distribution $\delta_y$), so the first term is $\sup_y S(\phi(y),y) = \sup_y G(\phi(y))$.
  So worst-case loss is bounded if and only if $G$ is.
\end{proof}

\begin{lemma} \label{lem:msr-domain-wcl-btb}
  On an unbounded domain, no expectation market satisfies WCL, assuming the initial market prediction lies in the interior of $\R$.
\end{lemma}
\begin{proof}
  
  Let $r$ be the market's initial starting point.
  We assume that $r$ is in the interior of the convex hull of $\{\phi(y) : y \in \Y\}$, and in particular lies in an $\epsilon$-ball contained in the interior.
  The idea is to pick a ``direction'' and consider a sequence of outcomes that are farther and farther away.
  If the final market state is some distance in that direction, then loss will be unbounded.

  To formalize this, let $y_1,y_2,\dots$ be a sequence with $\|\phi(y_i) - r\|$ positive and increasing without bound.
  Such a sequence must exist by unboundedness of the domain.
  Consider the associated sequence of size-$\epsilon$ vectors $\{v_i = \epsilon \frac{\phi(y_i)-r}{\|\phi(y_i)-r\|} : i=1,2,\dots\}$.
  As a sequence in a compact set (the $\epsilon$-sphere), it has a convergent subsequence with a limit $v$.
  This is the ``direction''.
  By monotonicity of strictly convex functions, we have $\gamma \defeq v \cdot (dG_{r+v} - dG_r) > 0$.
  Let $K \subseteq \mathbb{N}$ be the indices of this subsequence.

  Now consider the following sequence of markets closing states and outcomes, indexed by $i \in K$.
  Each market $i$ has initial state $r$ and final state $r + v$.
  The outcome is $Y = y_i$.
  The initial state is valid because $r$ was fixed at the beginning of the proof, and $r + v$ is a valid final closing state because it lies on the $\epsilon$-ball around $r$, assumed to lie in the report space because $r$ was in the interior.
  The worst-case loss is at least the following quantity, which is then rearranged:
  \begin{align*}
    \text{Loss} &= S(r+v,\phi(y_i)) - S(r,\phi(y_i))  \\
                &= G(r+v) - G(r) + dG_{r+v}\cdot(\phi(y_i) - r - v) - dG_r \cdot(\phi(y_i) - r)  \\
                &= \left[G(r+v) - G(r) - dG_r\cdot (r+v-r)\right] + \left(dG_{r+v} - dG_r\right)\cdot\left(\phi(y_i) - r - v\right)  \\
                &\geq \left(dG_{r+v} - dG_r\right)\cdot\left(\phi(y_i) - r - v\right)
  \end{align*}
  using the definition of subgradient to conclude that the bracketed term is at least zero.
  Now we divide both sides by $\|\phi(y_i) - r\|$ and recall our definition of $v_i$:
  \begin{align*}
    \frac{\text{Loss}}{\|\phi(y_i) - r\|}
                &\geq  \epsilon \left[\left(dG_{r+v} - dG_r\right) \cdot v_i \right] - \left[\left(dG_{r+v} - dG_r\right) \cdot \frac{v}{\|\phi(y_i) - r\|} \right].
  \end{align*}
  The first term on the right side converges to $\epsilon \gamma$ because $v_i \to v$; the second term is equal to $\gamma/\|\phi(y_i) - r\|$ which converges to zero.
  The ratio on the left therefore either diverges or converges to some constant larger than $\epsilon \gamma$, and in either case, worst-case loss is unbounded.
\end{proof}

\begin{proposition*}[Proposition \ref{prop:msr-expectation-btb}]
  The linear property $\Gamma(p) = \E_p \phi(Y)$ has an expectation market satisfying BTB if its domain $\conv(\phi(\Y))$ is bounded, in which case the market defined by any differentiable $G$ satisfies BTB.
\end{proposition*}
\begin{proof}
  Suppose $G$ is differentiable and its domain is bounded, e.g. $\|x\| \leq B$ for all $x$.
  Let $x^0$ be the market state and consider belief $\mu$.
  By monotonicity of strictly convex functions, any trade $x' = \alpha \mu + (1-\alpha)x^0$ has strictly positive expected score.
  Meanwhile, the worst-case score for any trade from $x^0$ to $x'$ is
  \begin{align*}
    &\sup_x G(x^0) - G(x') + dG_{x'}\cdot(x-x') - dG_{x^0}\cdot(x-x^0) \\
    &= \sup_x G(x^0) - G(x') + x\cdot\left(dG_{x'} - dG_{x^0}\right) - x'\cdot dG_{x'} - x^0\cdot dG_{x^0} \\
    &\leq B\|dG_{x'} - dG_{x^0}\| + O(\|x^0 - x'\|)
  \end{align*}
  where both terms can be made arbitrarily small with $\alpha \to 0$, $x' \to x^0$: $G$ is continuous, and it is a convex differentiable function so $dG$ is as well.
\end{proof}

\begin{theorem*}[Theorem \ref{thm:linear-finiteY-pn}]
  On a finite outcome space, i.e. $|\Y| < \infty$, for any linear property there exists an expectation market satisfying PN, WCL, and BTB.
\end{theorem*}
\begin{proof}
  We utilize Theorem \ref{thm:cost-pn}, which says that if a cost-function based market is \emph{open} (Definition \ref{def:open-market}), then it satisfies PN.

  The primary examples are \emph{exponential-family} markets~\citep{abernethy2014information}, where
  \[ C(q) = \log \sum_{y \in \Y} \exp(q \cdot \phi(y)) , \]
  the ``log-partition'' function.
  We recall the key facts behind the construction and refer the reader to \citet{abernethy2014information}.
  One interprets $q\cdot \phi(y)$ as a ``weight'' on $y$ and define the (exponential-family) distribution $p \in \interior(\Delta_{\Y})$ with $p_y = \frac{\exp(q \cdot \phi(y))}{\sum_{y\in\Y} \exp(q\cdot \phi(y))}$.
  One obtains
  \[ \nabla C(q) = \sum_y p_y \phi(y) \]
  so the set of gradients is exactly the interior of $\conv(\phi(\Y))$, i.e. the market is open.

  By Theorem \ref{thm:cost-pn}, this implies TN.
  The convex conjugate $G(\mu)$ is bounded (equaling the negative entropy of $p$); this implies WCL by Lemma \ref{lem:msr-g-bound-wcl}.
  Finally, it and $C$ are both strictly convex and differentiable.
  This implies BTB by Proposition \ref{prop:msr-expectation-btb}.
\end{proof}

\begin{lemma}
  \label{lem:cost-function-prices}
  Given $\phi:\Y\to\reals^k$ such that $\Xc := \conv(\phi(\Y))$ is full-dimensional in $\reals^k$, let $C:\reals^k\to\reals$ be convex with subgradients $\partial C(\reals^k) \subseteq \interior(\Xc)$.  Then
For all $q,v\in \reals^k$, $\max_{y\in\Y} v\cdot\phi(y) > C(q+v) - C(q)$.  
\end{lemma}
\begin{proof}
  As $\partial C(q') \subseteq \interior(\Xc)$ for all $q'\in\reals^k$, in particular $dC(q')\cdot v < \max_{x\in \Xc} x\cdot v = \max_{y\in\Y} v\cdot\phi(y)$~\cite[Prop A.2.4.6]{urruty2001fundamentals}, where $dC$ is a selection of subgradients of $C$.  By~\cite[Thm B.4]{frongillo2014general}, the function $t\mapsto dC(q+tv)\cdot v$ is monotone and therefore integrable, and thus $C(q+v) - C(q) = \int_{t=0}^1 dC(q+tv)\cdot v \,dt
  < \max_{y\in\Y} v\cdot\phi(y)$.
\end{proof}

\begin{theorem} \label{thm:cost-pn}
  Cost-function-based markets satisfy TN if and only if they are open.  Moreover, if they satisfy TN, the satisfy PN.
\end{theorem}
\begin{proof}
  Suppose a cost-function-based market is open.
  We will show that it satisfies not only TN but PN.
  Consider a non-empty set of bundles purchased $v_i = q_i'-q_i$ for $1\leq i\leq m$, for a total cost of $c = \sum_i C(q_i')-C(q_i)$, and let $v = \sum_{i=1}^m v_i$ be their sum.  Clearly, to neutralize this position from the current market state $q$, the trader must sell $v$, for a cost of $C(q-v)-C(q)$.  After this trade, the trader's total contract is $c\ones - (C(q-v)-C(q))\ones$, so to establish PN, we need only show $c - C(q-v) + C(q) > \inf_{y\in\Y} [c + (-v)\cdot\phi(y)]$.  Subtracting $c$ from both sides and then negating, the rest follows from observing $\sup=\max$ as $\Y$ is finite, and applying Lemma~\ref{lem:cost-function-prices}.

  Now suppose a cost-function-based market satisfies TN; we will show it is open.
  By assumption, such a market has a differentiable $C$ with $\clo\left(\{\nabla C(q) : q \in \reals^k\}\right) = \conv(\phi(\Y)) =: \Xc$.
  Suppose the market is not open; this implies that $\nabla C(q)$ lies on the boundary of $\Xc$ for some $q$.
  As $\Xc$ is a convex polytope, $\nabla C(q)$ must lie on an exposed face of $\Xc$, so let $v\in\reals^k$ be a direction exposing that face, meaning $\max_{x\in\Xc} x\cdot v = \nabla C(q)\cdot v$~\cite[Sec A.2.4]{urruty2001fundamentals}.

  Now suppose the trader buys the bundle $(-v)$ at state $q$, and the market state has returned to $q$, which corresponds to $q_1=q$, $q_1'=q-v$, $q_2=q$.
  To satisfy TN, the trader must now purchase $v$, but we must additionally have $\inf[(C(q) - C(q-v))\ones + (C(q) - C(q+v))\ones] > \inf[(C(q)-C(q-v))\ones + (-v)\cdot \phi]$, which is equivalent to $C(q+v)-C(q) < \max_{y\in\Y} v\cdot\phi(y)$.
  By weak monotonicity~\cite[Thm 24.9]{rockafellar1997convex}, the map $t\mapsto \nabla C(q+tv)\cdot v$ is monotone increasing in $t$, and thus as $\nabla C(q+tv) \in \Xc$ and $\nabla C(q)\cdot v = \max_{x\in\Xc} x\cdot v$, we have $\nabla C(q+tv)\cdot v = \nabla C(q)\cdot v$ for all $t\geq 0$.
  But now we have $C(q+v) - C(q) = \int_{t=0}^1 \nabla C(q+tv) \cdot v dt = \nabla C(q) \cdot v = \max_{y\in\Y} \phi(y)\cdot v$, so the trader's minimum payoff has not increased, violating TN.
\end{proof}

\section{Median and Quantile Markets} \label{app:quantile}
Our setting in this section is as follows.  Let $\R=\Y=\reals$, so that reports and outcomes are both the real line.  For any $\alpha \in (0,1)$, we define the $\alpha$-quantile $q_\alpha$ of probability distribution $p$ on $\reals$ by $q_\alpha(p) = \{x : \lim_{x'\uparrow x}\Pr_p[Y\leq x'] \leq \alpha \leq \Pr_p[Y\leq x]\}$, which for continuous cumulative distribution functions is equivalent to the usual notion, $\Pr_p[Y\leq q_\alpha(p)] = \alpha$.  Of course, the median is simply $q_\alpha$ for $\alpha=1/2$.

Our first task will be to identify the full class of proper scoring rules eliciting the $\alpha$ quantile.  Several authors have proved such characterizations under slightly different assumptions on the class of probability distributions.  Our focus here is not on the technical details of the precise set of probability measures allowed, and in any case, the form of the scoring rules are nearly identical in all cases.
We give the characterization of Schervish et al..

\begin{theorem}[Schervish et al. 2013]
  Let $\P$ be a set of probability distributions on $\reals$ containing every distribution supported on $\{a,b\}$ for any $a,b\in\reals$.  Then scoring rule $S:\reals\times\reals\to\reals$ elicits $q_\alpha$ if and only if $\E_p S(r,Y) < \infty$ for all $r\in\reals,p\in\P$ and
  \begin{equation}
    \label{eq:msr-quantile}
    S(r, y) - S(y,y) = (\alpha - \ones\{r \geq y\})(g(r)-g(y))  = \begin{cases}
  \alpha (g(r) - g(y)) & r < y\\
  (\alpha-1)(g(y) - g(r)) & r \geq y
\end{cases}~,
  \end{equation}
  for some strictly increasing function $g$.
\end{theorem}
When $\alpha=1/2$, as $r<y \iff g(r)<g(y)$, the form~\eqref{eq:msr-quantile} simplifies to $S(r,y) = -\tfrac 1 2 |g(r)-g(y)|$.
The usual absolute loss then follows as the special case when $g$ is the identity function, $g(r) = r$ (though note we must negate the loss to obtain a score).

\begin{lemma}
  \label{lem:msr-quantile-min}
  For quantile market $S(r,y) = (\alpha - \ones\{r \geq y\})(g(r)-g(y))$, we have
  $\min_y F(r',y|r) = \alpha(g(r)-g(r')$ when $r'>r$, achieved at any $y \leq r$, and $\min_y F(r',y|r) = (\alpha-1)(g(r)-g(r'))$ when $r'<r$, achieved at any $y > r$.
\end{lemma}
\begin{proof}
  We proceed by a simple calculation.

  \begin{align*}
    F(r',y|r) &= (\alpha - \ones\{r' \geq y\})(g(r')-g(y)) - (\alpha -
    \ones\{r \geq y\})(g(r)-g(y))
    \\
    &= \begin{cases}
      \alpha (g(r') - g(y)) - \alpha(g(r)-g(y)) & r,r' < y\\
      \alpha (g(r') - g(y)) - (\alpha-1)(g(r)-g(y)) & r < y \leq r'\\
      (\alpha-1) (g(r') - g(y)) - \alpha(g(r)-g(y)) & r' < y \leq r\\
      (\alpha-1) (g(r') - g(y)) - (\alpha-1)(g(r)-g(y)) & y \leq r,r'
    \end{cases}
    \\
    &= \begin{cases}
      \alpha (g(r') - g(r)) & r,r' < y\\
      \alpha (g(r') - g(r)) + (g(r)-g(y)) & r < y \leq r'\\
      \alpha (g(r') - g(r)) + (g(y)-g(r')) & r' < y \leq r\\
      (\alpha-1) (g(r') - g(r)) & y \leq r,r'
    \end{cases}
    \\
    &\geq \begin{cases}
      \alpha (g(r') - g(r)) & r,r' < y\\
      \alpha (g(r') - g(r)) + (g(r)-g(r')) & r < y \leq r'\\
      \alpha (g(r') - g(r)) + (g(r)-g(r')) & r' < y \leq r\\
      (\alpha-1) (g(r') - g(r)) & y \leq r,r'
    \end{cases}~.
  \end{align*}
\end{proof}

\begin{theorem*}[Theorem \ref{thm:quantile-tn}]
  No SRM for any $\alpha$-quantile satisfies WN.
\end{theorem*}
\begin{proof}
  Suppose the trader had previously moved the market state from $r_1\to r_2$ and the current market state is at $r_3$, with $r_3 < r_1 < r_2$.  Clearly taking $r_4=r_3$ has no effect, so we have two cases.  If $r_4 > r_3$, then consider any $y < r_3$.  (Intuitively, this is the ``bad'' outcome as both trades moved the market state away from $y$.)  By Lemma~\ref{lem:msr-quantile-min}, we see that $F(r_2,y|r_1) = \min_{y'}F(r_2,y'|r_1)$ already for this $y$, and by the same Lemma, $F(r_4,y|r_3) = (\alpha-1)(g(r_4)-g(r_3)) < 0$, so the minimum overall payoff has decreased.  Similarly, if $r_4 < r_3$, then we take any $y$ with $r_3 < y < r_1$.  (Again, this $y$ is chosen so that both trades move the market state away from $y$.)  Then by Lemma~\ref{lem:msr-quantile-min} we have we have $F(r_2,y|r_1) = \min_{y'}F(r_2,y'|r_1)$ and $F(r_4,y|r_3) = \alpha(g(r_4)-g(r_3)) < 0$, so again we have decreased the minimum overall payoff.  As taking $r_4 = r_3$ does not change the minimum overall payoff, there is no trade which increases it, and thus the market $F$ does not satisfy WN.
\end{proof}

\begin{theorem*}[Theorem \ref{thm:quantile-wcl}]
  For all $\alpha\in(0,1)$, there is an SRM for the $\alpha$-quantile satisfying WCL.
\end{theorem*}
\begin{proof}
  PI and IC are satisfied by choosing a scoring rule of the form $S(r,y) = (\alpha - \ones\{r \geq y\})(g(r)-g(y))$, as above.
  Now take $g$ to be any bounded, positive, monotone increasing function, e.g. the sigmoid function $g(r) = e^r/(1+e^r)$.
  This $S$ still elicits the $\alpha$-quantile because the payoff for every $r$ and $y$ is the same, but $S$ is now bounded above and below.
  By the telescoping property of the SRM, the worst-case loss is of the form $S(r,y) - S(r',y)$, and this is bounded for all $r,y$.
\end{proof}

\begin{theorem*}[Theorem \ref{thm:quantile-btb}]
  If distributions in $\P$ do not contain point masses, then any $\alpha$-quantile market satisfies BTB.
\end{theorem*}
\begin{proof}
  Let $y$ be the current market prediction.
  Let a budget $\epsilon$ be given.
  By continuity of $g$, there is a report $r^* > y$ with $g(r^*) \leq g(y) + \epsilon$.
  Any report $r \in (y, r^*]$ has worst-case budget of at most $\epsilon \cdot \max\{\alpha, 1-\alpha\} \leq \epsilon$.

  Suppose $y$ is a $c$-quantile according to a trader's belief, assumed to be a continuous CDF, with $c < \alpha$ without loss of generality.
  Pick a report $r \in (y, r^*]$ whose quantile according to the belief is some $b \in (c, \alpha)$ (such a report must exist by continuity).
  Let $\epsilon' = g(r) - g(y)$, in $(0,\epsilon]$ by assumption.
  We can divide expected payoff into two cases.
  If $Y > r$, the net payout for the trade is $\alpha \epsilon'$.
  If $Y \leq r$, the net payout for the trade is at least $-(1-\alpha)\epsilon'$, which is the minimum possible for any $Y \leq r$.
  So the expected payout according to this belief is at least $ (1-b)\alpha \epsilon' - b (1-\alpha) \epsilon' = \epsilon' \left[ \alpha - b \right] > 0 $.
\end{proof}

\section{Expectiles} \label{app:expectiles}

Let $\R = \Y = \reals$.   The $\tau$-expectile is defined as the solution $x=\mu_\tau(p)$ to the equation $\E_p\left[|\ones_{x\geq Y}-\tau|(x-Y)\right]=0$.  Again according to Gneiting~\citep{gneiting2011making}, subject to mild conditions, a scoring rule elicits $\mu_\tau$ if and only if it takes the form $S(r,y) = -|\ones\{y\leq r\} - \tau|(g(y)-g(r)-dg_r(y-r))$ where $g:\R\to\reals$ is strictly convex with selection of subgradients $dg$.  We will write $D_g(y,r) = g(y)-g(r)-dg_r(y-r)$ for the corresponding Bregman divergence, keeping $dg$ implicit.

\begin{lemma}
  \label{lem:breg-div-quasi-convex}
  The Bregman divergence $D_g(y,x)$ is continuous and convex in $y$, and quasi-convex in $x$.
\end{lemma}

\begin{lemma}
  \label{lem:expectile-monotone}
  Suppose $r'>r$.  If $F$ is IC for an expectile, $F(r',y|r)$ is continuous and strictly monotone increasing in $y$.  Moreover, $(2\tau-1)F(r',y|r)$ is convex in $y$.
\end{lemma}
\begin{proof}
  First we show monotonicity.  By the definition of $S$, we have two cases:

  Case 1: $y \in [r,r']$.  By Lemma~\ref{lem:breg-div-quasi-convex}, $D_g(y,x)$ is quasiconvex in $y$, and takes a minimum at $y=x$, it must be increasing in $y$ for $y>x$ and decreasing for $y < x$.  Putting these together gives the claim for this case.

  Case 2: $y > r'$ or $y < r$.  Both reduce to $F(r',y|r) = \alpha (D_g(y,r) - D_g(y,r'))$ for $\alpha > 0$.  Simplifying and grouping constants, we have $F(r',y|r) = \alpha (dg_{r'} - dg_r) y + C$, and as $dg_x$ is strictly monotone increasing in $x$, the coefficient of $y$ is strictly positive, and hence $F$ is increasing.

  For continuity, we simply note that $D_g$ is continuous in $y$ in Case 1, affine in $y$ and thus continuous in Case 2, and one easily checks that the two cases coincide when $y=r$ or $y=r'$, as $D_g(r,r)=D_g(r',r')=0$.  Convexity follows from the fact that $F$ is linear in Case 2, and in Case 1 reduces to $(2\tau - 1)g(y) + C_1 y + C_2$ for constants $C_1,C_2$; multiplying again by $(2\tau-1)$ gives a positive coefficient to $g(y)$, which is convex.
\end{proof}

\begin{theorem}
  \label{thm:expectile-wn}
  Any  expectile market with differentiable $g$ satisfies WN.
\end{theorem}
\begin{proof}
  Suppose $r_1' > r_1$ without loss of generality, and let $r_2$ be arbitrary.  By continuity and monotonicity of the derivative $g'(\cdot)$ of $g$, there exists $r_2'<r_2$ such that $g'(r_2') - g'(r_2) = g'(r_1) - g'(r_1')$; we show that this $r_2'$ satisfies the conditions of WN.  By Lemma~\ref{lem:expectile-monotone}, we see that $F(r_1',y|r_1) + F(r_2',y|r_2)$ is constant for $y<\min\{r_1,r_2'\}$ and $y > \max\{r_1',r_2\}$, and by continuity, must be bounded in between; as $F(r_1',y|r_1)$ was unbounded before, we have satisfied WN.  
\end{proof}

\includegraphics{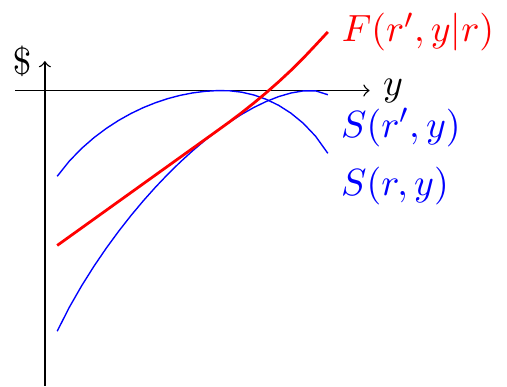}
\qquad
\includegraphics{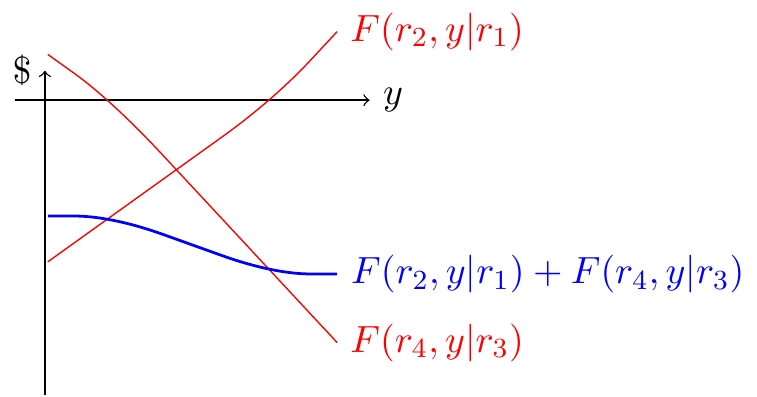}
%
%

%
%
%
%
%
%
%
%
%
%
%
%
%
%
%
%
%
%
%
%
%
%
%
%
%
%
%
%
%

\end{document}